\documentclass{llncs}

\usepackage[english]{babel}
\usepackage[utf8]{inputenc}
\usepackage{framed}
\usepackage{paralist}
\usepackage{url}
\usepackage{graphicx}
\usepackage{wrapfig}
\usepackage{LobsterTwo}
\usepackage{fetamont}
\usepackage[T1]{fontenc}
\usepackage{arydshln}  
\usepackage{bbm}
\usepackage{amssymb}     
\usepackage{amsmath}
\usepackage[all]{xy}
\usepackage{changepage}
\usepackage{relsize}
\usepackage{scalerel}
\usepackage{stackengine}
\usepackage[new]{old-arrows}
\usepackage[font=footnotesize]{subfig}

\usepackage{bussproofs}
\EnableBpAbbreviations

\usepackage{alltt}
\usepackage{bbm}
\usepackage{stmaryrd}
\usepackage{amsfonts}
\usepackage{mathrsfs}

\usepackage{xspace}
\usepackage{xcolor}
\usepackage{framed,color}
\usepackage{todonotes}
\definecolor{AirForceBlue}{HTML}{5D8AA8}

\definecolor{Alizarin}{HTML}{E32636}

\def\qed{$\quad\Box$}
\def\Nat{\mathbbm{N}}

\def \<{\left\langle}
\def \>{\right\rangle}
\def \({\left(}
\def \){\right)}
\def \op{{\sf op}}

\def \setof#1#2{\setbox1\hbox{$#1$}
                \setbox2\hbox{$#2$}
                \ifdim \ht1 > \ht2
                   \left \{ \left . \, #1 \, \right \vert \, #2 \, \right \}
                \else
                   \left \{ \, #1 \, \left \vert \, #2 \, \right . \right \} 
                \fi}

\newcommand\restr[2]{{
  \left.\kern-\nulldelimiterspace 
  #1 
  \vphantom{\big|} 
  \right|_{#2} 
  }}

\def \halfthinspace{\relax\ifmmode\mskip.5\thinmuskip\relax\else\kern.8888em\fi}

\def \aconv#1{\setbox13\hbox{$#1$}\ifdim\wd13<12pt\breve{#1}\else{\(#1\)}\breve{\ }\fi}

\def \pconv#1{\setbox13\hbox{$#1$}\ifdim\wd13<10pt\stackrel{\smile}{#1}\else{\(#1\)}^{\smile}\fi}

\def \qed{\ifmmode\rule{5pt}{5pt}\else{\nobreak\hfil\penalty50\hskip1em\null\nobreak\hfil\rule{5pt}{5pt}\parfillskip=0pt\finalhyphendemerits=0\endgraf}\fi}

\newcommand\nattrans{{\ \Rightarrow\ }}
\newtheorem{fact}{Fact}[section]

\newcommand\HeteroGenius{{{\ffmfamily Hetero}{\LobsterTwo Genius}}}

\title{Integrating deduction and model finding in a language independent setting}
\author{
         Carlos G. Lopez Pombo\inst{1,2} 
\and Agust\'{\i}n Eloy Martinez Su\~{n}\'{e}\inst{1}
}
\authorrunning{Lopez Pombo et. al.}

\institute{
        CONICET-Universidad de Buenos Aires. Instituto de Investigaci\'{o}n en Ciencias de la Computaci\'{o}n (ICC)
\and Department of Computing, FCEyN, Universidad de Buenos Aires 
\email{\{clpombo,aemartinez\}@dc.uba.ar} 
}

\begin{document}

\maketitle


\begin{abstract}
Software artifacts are ubiquitous in our lives being an essential part of home appliances, cars, cel phones, and even in more critical activities like aeronautics and health sciences. In this context software failures may produce enormous losses, either economical or, in the extreme, in human lives. Software analysis is an area in software engineering concerned on the application of different techniques in order to prove the (relative) absence of errors in software artifacts. 

In many cases these methods of analysis are applied by following certain methodological directives that ensure better results. In a previous work we presented the notion of satisfiability calculus as a model theoretical counterpart of Meseguer's proof calculus, providing a formal foundation for a variety of tools that are based on model construction.

The present work shows how effective satisfiability sub-calculi, a special type of satisfiability calculi, can be combined with proof calculi, in order to provide foundations to certain methodological approaches to software analysis by relating the construction of finite counterexamples and the absence of proofs, in an abstract categorical setting.
\end{abstract}


\section{Introduction and motivation}
\label{intro}
Software artifacts are ubiquitous in our lives being an essential part of home appliances, cars, cel phones, and even in more critical activities like aeronautics and health sciences. In this context software failures may produce enormous losses, either economical or, in the extreme, in human lives. In the latter context, compliance with certain quality standards might be mandatory. Software analysis is an area in software engineering concerned on the application of different techniques in order to prove the (relative) absence of errors in software artifacts. Software analysis, usually (if not always) requires that a formal description of the behaviour of the system, known as its specification, must be available. Then, it is possible to check wether a given property of interest follows from it.\\

Logics have often been used as formal systems suitable for the specification of software artifacts. Moreover, logical specifications, due to its formal nature, have contributed towards the application of sound verification techniques. Several formalisms have been developed to cope with these aspects and most of them are effective in describing some particular characteristics of software systems. Among many examples, one can mention:
\begin{itemize} 
\item linear-time temporal logics, both propositional \cite{pnueli:ieee-focs77,pnueli:tcs-13_1} and first-order \cite{manna95}, for describing properties about single executions, 
\item branching-time temporal logics \cite{benari:acm-sigplan-sigact81}, for describing properties of the class of potential executions from a given state of the system, 
\item $\text{CTL}^*$ \cite{emerson:jacm-33_1} as a compromise sitting in between linear-time and branching-time logics, 
\item the many versions of dynamic logics \cite{fischer:stoc77,harel00}, formalising properties of (imperative) sequential programs, 
\item dynamic linear temporal logic \cite{henriksen:apal-96_1_3}, for embedding the dynamic characterisation of programs within the temporal operators linear temporal logics, 
\item higher-order logics \cite{vanbenthem:hlfcs83} for capturing several notions like higher-order functions, 
\item equational logic \cite{henkin:amm-84_8} for capturing abstract data types \cite{goguen:cgprds75}, and the list continues.
\end{itemize}

The usefulness of logic in software design and development has always been some form of universal, and well-documented, truth in computer science. From Turing's initial observations, already in \cite{turing:plms-s2-42_1}, stating that ``the state of progress of a computation'' is completely determined by a single expression, referred to as ``state formulae''; through Hoare's ``Axiomatic Basis for Computer Programming'' \cite{hoare:cacm-26_1}
, Burstall's ``Proving Properties of Programs by Structural Induction'' \cite{burstall:tcj-12_1}
\ and Floyd's ``Assigning meaning to programs'' \cite{floyd+:amsam66}
, in the 60's; to Parnas' ``A Technique for Software Module Specification with Examples'' \cite{parnas:cacm-15_5}
\ and Dijkstra's ``Guarded Commands, Nondeterminacy and Formal derivation of Programs'' \cite{dijkstra:cacm-18_8}
, in the 70's; just to start with.

Formalising a software artefact by resorting to a logical language requires people involved in the development process to agree on how such descriptions are to be interpreted and understood, so that the systems can be built in accordance to its description and, therefore, have the expected behaviour. Semantics plays a central role in this endeavour as it is a way of substituting the perhaps drier and more esoteric forms provided by syntactic descriptions, by the more intuitive modes of understanding appealing to some naive form of set theory, but the metalogical relation between these two inherent aspects of logical languages, syntax and semantics, might reveal certain discordance with significant impact regarding the appropriateness of the use of one or the other.\\

Back in the second half of the 70's, when formal software specification was booming due to the appearance of Abstract Data Types \cite{liskov:afips72} as the main driver in programming, logic ended up in a privileged position in computer science. Tom Maibaum, and a group of Brazilian logicians leaded by Paulo Veloso, realised that the use abstract semantics for interpreting data types specifications, and thinking about their representations, was not appropriate. In such a setting, either implicit or explicit in many different approaches at that time \cite{hoare:dahl72,hoare:ai+-1_4,hoare:ijcis-4_2,liskov:acmsigplannot-9_4,liskov:acmsigplannot-10_6,liskov:yeh77,guttag:cacm-20_6,guttag:cacm-21_12}, not every model is reachable (i.e., those built from objects that cannot be named by syntactic terms) and, even among those that are, there are some that are not minimal in the sense that they might satisfy properties other than those provable from the axioms. They also observed that such minimality was a direct consequence from requiring models to be initial, as the ADJ group, leaded by Joseph Goguen, was pushing forward for algebraic specification languages \cite{goguen:cgprds75,goguen:yeh78+}.

Maibaum's commitment with this view on semantics in computer science is best understood in his words:
\begin{quotation}
\emph{``It cannot be emphasized enough that it is this concept of initiality which gives rise to the power of the method. Once the requirement of initiality is relaxed, many of the results and proof methods associated with the concept disappear.''}
\end{quotation}

Such a strong standpoint motivated him, and his colleagues, to pursue a research program \cite{carvalho:CS-80-22}\footnote{A shorter version was published later in \cite{carvalho:ics82}.} aiming at producing a theory of initiality for first order predicate logic, thus, reproducing ADJ's results about equational logic, but for a logical language expressive enough to capture the behaviour of complex data structures and, therefore, more suitable to software systems specification. 

In \cite{maibaum:tapsoft97} Maibaum cites \cite{ehrich:jacm-29_1}, a seminal paper by Hans-Dieter Ehrich, as an important source for the motivations behind the foundational notions of the approach, outstandingly summarised by the former as:
\begin{quotation}
\emph{``Why did correct implementations \underline{not} compose correctly in all cases? What were the engineering assumptions on which the technical developments were based? What results were dependent on the formalism used (some variant of equational logic or FOL\footnote{As usual, FOL refers to the fisrt order predicate calculus, or first order logic.} or whatever) and which were 'universal'?''}
\end{quotation}

Such questions reveal that Maibaum was committed with the provision of a tool for practical software development, and not just the capricious definition of a well-founded theory of software specification. Truth be said, while Ehrich and Maibaum had the same motivations, the work of the former was still focussed on the equational axiomatisation of abstract data types, while the latter's program was moving further to a more expressive logical framework, first order predicate calculus. Much of the progress was documented as technical reports of the institutions to which the authors were affiliated at the time \cite{maibaum:TRun81,maibaum:TRun83,maibaum:TRun83+} and later published in \cite{maibaum:fsttcs84,maibaum:tapsoft85}. While the main objectives were preserved, the specification language adopted was $L_{{\omega_1}, \omega}$ \cite[Sec.~11.4]{karp64} (i.e., first-order predicate logic with equality, admitting conjunctions and disjunctions of denumerable sets of formulae), due to the technical need of what the authors call ``namability axioms'' (i.e., formulae of the shape $(\forall x: T)(\bigvee_{n \in \text{Name} (T)} x =_T n )$, for each sort $T$) for ensuring that any model, from the traditional model theoretic perspective, is reachable.\\

Also in the mid 80's, Goguen and Burstall introduced \emph{Institutions} \cite{goguen:cmwlp84} as a categorical framework aiming at providing a formal framework for defining the notion of logical system, from a model-theoretic perspective. 
Thanks to its categorical presentation, Institutions were key in revealing more elaborate views on software specifications (see, for example, those focussed on compositional specifications \cite{tarlecki:ctcp86,sannella:ic-76_2-3,duran:tcs-309_1-3}, software development by refinement of specifications \cite{sannella:ai+-25_3,sannella:calp92} and heterogenous specifications \cite{tarlecki:sadt-rtdts95,arrais:sadt-rtdts95,mossakowski:icalp1996,tarlecki:gabbay00,bernot:amast96,mossakowski:wadt02}, among many others) and elegantly supporting many features pursued in Maibaum's endeavour, like well-behaved composition mechanisms, semantically consistent parameterisation procedures, etc. What for some researchers, like the authors of this work, was interpreted as a source of mathematical beauty, by enabling a flexible way of formalising the notion of satisfaction associated to an arbitrary class of models of interest, for Maibaum was a deal breaker, Goguen's departure from \underline{mandatory} initial semantics ended the tactical alliance forever.\\

The theory of Institutions evolved in many directions being \emph{General logics} \cite{meseguer:lc87} of specific interest to this work. In it, Meseguer, complemented the definition of a logical system from a semantical point of view, provided by institutions, with a categorical characterisation for the notions of entailment system (also referred to as $\pi$-institutions by Fiadeiro and Maibaum \cite{fiadeiro:icdcwtfm93}). Such a syntactical view of a logical system allowed for a more balanced definition of the notion of \emph{Logic}, integrating both formalisations. Thereafter, the author sharply observes that there might be ``a reasonable objection'' regarding the abstractness of the notion of entailment as it hides out the internal structure the entailment relation (i.e., the mere notion of proofs). Such an abstract view of the notion of entailment might be considered more a virtue that a defect, as it enables the possibility of defining many plausible proof systems, for a single entailment relation (for example, the same abstract entailment relation between sets of formulae and formulae can be defined as a Hilbert-style calculus, a natural deduction calculus or by other means). Then, Meseguer completes his view on logics by proposing categorical characterisations of the notions of \emph{Proof calculus}, \emph{Proof subcalculus} (the restriction of a proof calculus obtained from identifying the subcategory of theories and the subset of conclusions that can be drawn from that theory), and \emph{Effective proof subcalculi} (a proof subcalculus whose proving method is guarantied to be effective by requiring sentences, axioms, conclusions and proofs to be structured in a space \cite{schoenfield71}).

The reader may have noticed that Meseguer's viewpoint induced a new, and further, imbalance towards syntax leaving the semantic aspects of a logical system too abstract with respect to the syntactic ones.\\

Model theoretic-based reasoning techniques constitute an important stream of research in logic; in particular, these methods play an important role in automated software validation and verification. The origin of this type of logical reasoning tools can be traced back to the works of Beth \cite{beth59,beth:hintikka69}, Herbrand \cite{herbrand:goldfarb71} and Gentzen \cite{gentzen:szabo69}; Beth's ideas were used by Smullyan to formulate the tableau method for first-order predicate logic \cite{smullyan95}. Herbrandt's and Gentzen's work inspired the formulation of resolution systems presented by Robinson \cite{robinson:jacm-12_1}. Methods like those based on resolution and tableaux are strongly related to the semantics of a logic; therefore, they can often be used as the mechanics behind the construction of models. Such a use is not possible in \textit{pure} deductive methods, such as natural deduction or Hilbert systems, as formalized by Meseguer.

In \cite{lopezpombo:fi-166_4}, Lopez Pombo and some colleagues, forced Maibaum into stretching his views trying to reach common ground. This was done by introducing a framework, analogous to Meseguer's mechanisation of the notion of proof, but serving the purpose of formalising how models are constructed from the sentences they have to satisfy, and how these ``canonically'' defined structures relate to the abstract notion of satisfaction formalised by the institution for which the calculus is being defined. The formalisation of such a constructive view of model theory received the name of \emph{Satisfiability calculus} and was later extended, in analogy to Meseguer's presentation of proof systems, to \emph{Satisfiability subcalculus} and \emph{Effective satisfiability subcalculus}.\\

Formal methods are usually divided into two categories: heavyweight and lightweight. These names refer to the amount of mathematical expertise needed during the process of proving a given property. For many specification languages a lower degree of involvement equates to a higher degree of automation and, consequently, to less certainty regarding the satisfaction of the property. Thus lightweight formal methods, like Alloy \cite{jackson:acmtosem-11_2}, cannot usually be entirely trusted when dealing with models of mission critical systems. An alternative is the use of heavyweight formal methods, as for instance semi-automatic theorem provers \cite{owre:cav96,dowek:coq-users-guide,nipkow02}. Theorem provers also exhibit limitations, for instance, they usually require a high level of expertise and strong mathematical skills, that many times discourages their use. More modern analysis methodologies departed from the idea that either heavyweight or lightweight formal methods are applied, disregarding the relation between these tools. Our claim is that formally enforcing certain methodological directives as part of the process of software analysis produces better results in practice. An example of this is Dynamite \cite{frias:tacas07}. Dynamite is a theorem prover for Alloy in which the critical parts of the proof are assisted by the Alloy Analyzer with the aim of reducing both the workload and the error proneness introduced by the human interaction with the tool. Another use of model theoretic tools in relation to the use of theorem provers is the fact that they provide an efficient method for:
\begin{inparaenum}[a)]
\item gaining confidence in the hypothesis brought into a proof,
\item the elimination of superfluous formulae appearing in a sequent,
\item the removal of minor modelling errors, and even 
\item the suggestion of potential witnesses for existential quantifiers.
\end{inparaenum}

Another example of the synergy between deduction and model finding is Nitpick \cite{blanchette:itp10}, an automatic counterexample finder for Isabelle \cite{nipkow02}, implemented using Kodkod \cite{torlak:tacas07}. The integration between Nitpick and Isabelle provides the command \texttt{nitpick} \cite{blanchette:nitipick-users-guide} that, when it is applied, first translates the current sequent to a Kodkod problem consisting of:
\begin{inparaenum}[1)]
\item the type declarations (including bounds for the amount of atoms to be considered for each one), and
\item a first-order relational logic \cite{jackson:acm-sigsoft00} formula equivalent to the conjunction of the hypothesis and the negation of the thesis;
\end{inparaenum}
and then invoques the latter in order to find a model of such a formula.

Nitpick also provides an automatic mode in which the command is applied every time the user inputs a rule in Isabelle's command prompt.\\

In this work we aim at providing a formal link supporting the methodological interaction between bounded model checkers and theorem provers by connecting the concrete structures over which models and proofs are represented in their formalisation as effective satisfiability sub-calculi and proof calculi, respectively.\\

The paper is organized as follows. In Sec.~\ref{preliminaries} we present the definitions and results we will use throughout this paper. In Sec.~\ref{subsat} we present a categorical formalization of satisfiability sub-calculus and prove relevant results underpinning the definitions leading to the formalisation of a methodology for software analysis. In Sec.~\ref{case-study} we present examples in enough detail to exemplify the ideas. Finally, in Sec.~\ref{conclusions}, we present some conclusions drawn from the contribution and propose a direction in which this research can be continued.


\section{General logics: the category-theoretic formalisation of logical systems}
\label{preliminaries}

From now on, we assume the reader has a nodding acquaintance with basic concepts from category theory \cite{maclane71,fiadeiro05}. 
We mainly follow the notation introduced in \cite{meseguer:lc87}.

An \emph{entailment system} is defined by identifying a family of \emph{syntactic} consequence relations indexed by the elements of a category assumed to be the category of signatures. As usual, entailment relations are required to satisfy reflexivity, monotonicity\footnote{The theory of institutions and general logics focus on monotonic logical systems. The interested reader is referred to \cite{cassano:wadt2012} for a presentation of entailment systems for default logic, a well-known non-monotonic logical system introduced by Reiter in \cite{reiter:ai-13_1-2} aiming at the formalisation of defeasible logical reasoning.} and transitivity.

\begin{definition}[Entailment system \cite{meseguer:lc87}]
\label{entailment-system}
A structure $\< \mathsf{Sign}, \mathbf{Sen}, \{\vdash^{\Sigma}\}_{\Sigma \in |\mathsf{Sign}|} \>$ is said to be an \emph{entailment system} if it satisfies the following conditions:
\begin{itemize}
\item $\mathsf{Sign}$ is a category of signatures,
\item $\mathbf{Sen}: \mathsf{Sign} \to \mathsf{Set}$ is a functor. Let $\Sigma \in |\mathsf{Sign}|$; then $\mathbf{Sen}(\Sigma)$ returns the set of $\Sigma$-sentences, and
\item $\{\vdash^{\Sigma}\}_{\Sigma \in |\mathsf{Sign}|}$, where $\vdash^{\Sigma} \subseteq 2^{\mathbf{Sen}(\Sigma)} \times \mathbf{Sen}(\Sigma)$, is a family of binary relations such that for any $\Sigma, \Sigma' \in |\mathsf{Sign}|$, $\{\phi\} \cup \{\phi_i\}_{i \in  \mathcal{I}} \subseteq \mathbf{Sen}(\Sigma)$, $\Gamma, \Gamma' \subseteq \mathbf{Sen}(\Sigma)$, the following conditions are satisfied:
\begin{itemize}
\item reflexivity: $\{\phi\} \vdash^\Sigma \phi$,
\item monotonicity: if $\Gamma \vdash^\Sigma \phi$ and $\Gamma \subseteq \Gamma'$, then $\Gamma' \vdash^\Sigma \phi$,
\item transitivity: if $\Gamma \vdash^\Sigma \phi_i$ for all $i \in \mathcal{I}$ and $\{\phi_i\}_{i \in \mathcal{I}} \vdash^\Sigma \phi$, then $\Gamma \vdash^\Sigma \phi$, and
\item $\vdash$-translation: if $\Gamma \vdash^\Sigma \phi$, then for any morphism $\sigma: \Sigma \to \Sigma'$ in $\mathsf{Sign}$, $\mathbf{Sen}(\sigma)(\Gamma) \vdash^{\Sigma'} \mathbf{Sen}(\sigma)(\phi)$. 
\end{itemize}
\end{itemize}
\end{definition}

\begin{definition}[Theory presentations \cite{meseguer:lc87}]
Let $\< \mathsf{Sign}, \mathbf{Sen}, \{\vdash^{\Sigma}\}_{\Sigma \in |\mathsf{Sign}|} \>$ be an entailment system. Its category of theory presentations is $\mathsf{Th} = \< \mathcal{O}, \mathcal{A} \>$ such that:
\begin{itemize} 
\item $\mathcal{O} = \setof{\< \Sigma, \Gamma \>}{\Sigma \in
  |\mathsf{Sign}|\ \mbox{and}\ \Gamma \subseteq \mathbf{Sen}(\Sigma)}$, and
\item $\mathcal{A} = \setof{\sigma: \<\Sigma, \Gamma\> \to \<\Sigma', \Gamma'\>}{
\begin{array}{l}
\< \Sigma, \Gamma \>, \< \Sigma', \Gamma' \> \in \mathcal{O}, \sigma: \Sigma \to \Sigma' \in ||\mathsf{Sign}||, \\
\mbox{ for all } \gamma \in \Gamma, \Gamma' \vdash^{\Sigma'} \mathbf{Sen}(\sigma)(\gamma).
\end{array}
}$
\end{itemize}
In addition, if a morphism $\sigma: \<\Sigma, \Gamma\> \to \<\Sigma', \Gamma'\>$ satisfies $\mathbf{Sen}(\sigma)(\Gamma) \subseteq \Gamma'$, it is called \emph{axiom preserving}. By retaining those morphisms of $\mathsf{Th}$ that are axiom preserving, we obtain the subcategory $\mathsf{Th}_0$.\end{definition}

Note that, in the previous definition, the objects of $\mathsf{Th}$ are determined by a signature and a set of axioms, which are not necessarily closed under entailment. A theory is obtained from a theory presentations by requiring:
\begin{inparaenum}[1)]
\item the former to satisfy all the axioms appearing in the latter, and 
\item to be closed under entailment.
\end{inparaenum}

\begin{definition}[Closure under entailment]
Let $\< \mathsf{Sign}, \mathbf{Sen}, \{\vdash^{\Sigma}\}_{\Sigma \in |\mathsf{Sign}|} \>$ be an entailment system and $\<\Sigma, \Gamma\> \in |{\sf Th}|$. We define $^\bullet: 2^{\mathbf{Sen}(\Sigma)} \to 2^{\mathbf{Sen}(\Sigma)}$ as follows: $\Gamma^\bullet = \setof{\gamma}{\Gamma \vdash^\Sigma \gamma}$. This function is extended to elements of ${\sf Th}$, by defining it as follows: $\<\Sigma, \Gamma\>^\bullet = \<\Sigma, \Gamma^\bullet\>$. $\Gamma^\bullet$ is called the theory generated by $\Gamma$.
\end{definition}

Roughly speaking, an institution is an abstract formalisation of the model theory of a logic in such a way that the existing relations between signatures, sentences over a signature and models for a signature are made explicit. These aspects are reflected by introducing the category of signatures, defining two functors capturing the sets of sentences and the classes of models, the first one going from this category to the category $\mathsf{Set}$ and the second one going from this category to $\mathsf{Cat}$, and by requiring the satisfiability relation to remain invariant under signature change.

\begin{definition}[Institution \cite{goguen:cmwlp84}]
\label{institution}
A structure $\< \mathsf{Sign}, \mathbf{Sen}, \mathbf{Mod}, \{\models^{\Sigma}\}_{\Sigma \in |\mathsf{Sign}|} \>$ is said to be an \emph{institution} if it satisfies the following conditions:
\begin{itemize}
\item $\mathsf{Sign}$ is a category of signatures,
\item $\mathbf{Sen}: \mathsf{Sign} \to \mathsf{Set}$ is a functor. Let $\Sigma \in |\mathsf{Sign}|$, then $\mathbf{Sen}(\Sigma)$ is its corresponding set of $\Sigma$-sentences, 
\item $\mathbf{Mod}: \mathsf{Sign}^\op \to \mathsf{Cat}$ is a functor. Let $\Sigma \in |\mathsf{Sign}|$, then $\mathbf{Mod}(\Sigma)$ is its corresponding category of $\Sigma$-models,
\item $\{\models^{\Sigma}\}_{\Sigma \in |\mathsf{Sign}|}$ is a family of binary relations $\models^{\Sigma} \subseteq |\mathbf{Mod}(\Sigma)| \times \mathbf{Sen}(\Sigma)$, for all $\Sigma \in |\mathsf{Sign}|$
\end{itemize}
\noindent such that for all $\sigma: \Sigma \to \Sigma' \in ||\mathsf{Sign}||$, $\phi \in \mathbf{Sen}(\Sigma)$ and $\mathcal{M}' \in |\mathbf{Mod}(\Sigma')|$, the following $\models$-invariance condition holds:
$$\mathcal{M}' \models^{\Sigma'} \mathbf{Sen}(\sigma)(\phi)\quad \mbox{ iff }\quad \mathbf{Mod}(\sigma^\op)(\mathcal{M}') \models^{\Sigma} \phi\ .$$
\end{definition}
Intuitively, the last condition of the previous definition says that \emph{truth is invariant with respect to notation change}. Given $\langle \Sigma, \Gamma \rangle \in |\mathsf{Th}|$, $\mathbf{Mod}: \mathsf{Th}^\op \to \mathsf{Cat}$ is the extension of the functor $\mathbf{Mod}: \mathsf{Sign}^\op \to \mathsf{Cat}$ such that $\mathbf{Mod} (\langle \Sigma, \Gamma \rangle)$ denotes the full subcategory of $\mathbf{Mod} (\Sigma)$ determined by those models $\mathcal{M} \in |\mathbf{Mod} (\Sigma)|$ such that $\mathcal{M} \models^\Sigma \gamma$, for all $\gamma \in \Gamma$. The relation $\models^\Sigma$ between sets of formulae and formulae is defined in the following way: given $\Sigma \in |\mathsf{Sign}|$, $\Gamma \subseteq \mathbf{Sen} (\Sigma)$ and $\alpha \in \mathbf{Sen} (\Sigma)$, 
$$
\Gamma \models^\Sigma \alpha\quad \mbox{if and only if} \quad \mathcal{M} \models^\Sigma \alpha \mbox{, for all $\mathcal{M} \in |\mathbf{Mod} (\langle \Sigma, \Gamma \rangle)|$.}
$$
Now, from Defs.~\ref{entailment-system}~and~\ref{institution}, it is possible to give a definition of \emph{logic} by relating both its model-theoretic and proof-theoretic characterisations. In this respect, coherence between the semantic and syntactic entailment relations is required, reflecting the standard concepts of soundness and completeness of logical systems.

\begin{definition}[Logic \cite{meseguer:lc87}]\ \\
\label{logic}
A structure $\< \mathsf{Sign}, \mathbf{Sen}, \mathbf{Mod}, \{\vdash^{\Sigma}\}_{\Sigma \in |\mathsf{Sign}|}, \{\models^{\Sigma}\}_{\Sigma \in |\mathsf{Sign}|} \>$ is said to be a \emph{logic} if it satisfies the following conditions:
\begin{itemize}
\item $\< \mathsf{Sign}, \mathbf{Sen}, \{\vdash^{\Sigma}\}_{\Sigma \in |\mathsf{Sign}|} \>$ is an entailment system, 
\item $\< \mathsf{Sign}, \mathbf{Sen}, \mathbf{Mod}, \{\models^{\Sigma}\}_{\Sigma \in |\mathsf{Sign}|} \>$ is an institution, and 
\item the following \emph{soundness} condition is satisfied: for any $\Sigma \in |\mathsf{Sign}|$, $\phi \in \mathbf{Sen}(\Sigma)$, $\Gamma \subseteq \mathbf{Sen}(\Sigma)$: 
$$\Gamma \vdash^\Sigma \phi\quad \mbox{implies}\quad \Gamma \models^\Sigma \phi\ .$$
\end{itemize}
In addition, a logic is \emph{complete} if the following condition is satisfied: for any $\Sigma \in |\mathsf{Sign}|$, $\phi \in \mathbf{Sen}(\Sigma)$, $\Gamma \subseteq \mathbf{Sen}(\Sigma)$: 
$$
\Gamma \models^\Sigma \phi\quad\mbox{implies}\quad\Gamma \vdash^\Sigma \phi.
$$
\end{definition}

Definition~\ref{entailment-system} associates deductive relations to signatures. As already discussed, it is important to analyse how these relations are obtained. The next definition formalises the notion of proof calculus by associating a proof-theoretic structure to the deductive relations introduced by the definitions of entailment systems. 

\begin{definition}[Proof calculus \cite{meseguer:lc87}]\ \\
\label{proof-calculus}
A structure $\< \mathsf{Sign}, \mathbf{Sen}, \{\vdash^{\Sigma}\}_{\Sigma \in |\mathsf{Sign}|}, \mathbf{P}, \mathbf{Pr}, \pi \>$ is said to be a \emph{proof calculus} if it satisfies the following conditions:
\begin{itemize}
\item $\< \mathsf{Sign}, \mathbf{Sen}, \{\vdash^{\Sigma}\}_{\Sigma \in |\mathsf{Sign}|} \>$ is an entailment system,
\item $\mathbf{P}: \mathsf{Th}_0 \to \mathsf{Struct}_{PC}$ is a functor. Let $T \in |\mathsf{Th}_0|$, then $\mathbf{P}(T) \in |\mathsf{Struct}_{PC}|$ is the proof-theoretical structure of $T$\footnote{The reader should note that $\mathsf{Struct}_{PC}$ strongly depends on the structure needed to formalise the concept of proof for a specific proof calculus. For example, while in \cite[Ex.~11]{meseguer:lc87} the formalisation of natural deduction for first-order logic requires the use of multicategories \cite[Def.~10]{meseguer:lc87}, in \cite[\S3]{lopezpombo:phdthesis} the formalisation of the proof calculus for $\omega$-closure fork algebras with urelements \cite[Def.~7]{frias:relmics01+} (a variant of fork algebras \cite{haeberer:ifiptc291,frias02} with a reflexive and transitive closure operator) requires the use of strict monoidal categories \cite[Ch.~VII, \S1]{maclane71} whose monoid of objects is given by the (not necessarily finite) subsets  of the corresponding class of equations.},
\item $\mathbf{Pr}: \mathsf{Struct}_{PC} \to \mathsf{Set}$ is a functor. Let $T \in |\mathsf{Th}_0|$, then $\mathbf{Pr}(\mathbf{P}(T))$ is the set of proofs of $T$; the composite functor $\mathbf{Pr} \circ \mathbf{P}: \mathsf{Th}_0 \to \mathsf{Set}$ will be denoted by $\mathbf{proofs}$, and
\item $\pi: \mathbf{proofs} \nattrans \mathbf{Sen}$ is a natural transformation such that for each $T = \<\Sigma, \Gamma\> \in |\mathsf{Th}_0|$ the image of $\pi_T: \mathbf{proofs}(T) \to \mathbf{Sen}(T)$ is the set $\Gamma^\bullet$. The map $\pi_T$ is called the \emph{projection from proofs to theorems} for the theory $T$.
\end{itemize}
\end{definition}


The use of the category $\mathsf{Th}_0$ for indexing proof structures responds to a technical need. Whenever we relate two theories with a morphism, say $\sigma:\<\Sigma, \Gamma\> \to \<\Sigma', \Gamma'\>$, the previous definition imposes a need for extending that relation to proofs of the form $\pi:\emptyset \to \alpha \in |\mathbf{proofs}(\<\Sigma, \Gamma\>)|$. If theories are taken from $\mathsf{Th}$, we know that there exists $\pi':\emptyset \to \mathbf{Sen}(\sigma)(\alpha) \in |\mathbf{proofs'}(\<\Sigma', \Gamma'\>)|$, but there is no obvious way to obtain it from $\pi$. If theories are taken from $\mathsf{Th}_0$, this problem no longer exists as the proof $\pi'$ is obtained by applying exactly  the same proof rules, obtaining the same proof structure (recall  the inclusion $\mathbf{Sen}(\sigma)(\Gamma) \subseteq \Gamma'$ in the definition of $\mathsf{Th}_0$).

Meseguer's categorical formulation of a proof calculus is a means of providing structure for the abstract relation of entailment defined in an entailment system. The next definition provides a categorical formalisation of a satisfiability calculus as it was presented in \cite{lopezpombo+:wadt2012-f}. A satisfiability calculus is the formal characterization of a method for constructing models of a given theory, thus providing the semantic counterpart of that proof calculus. In the same way Meseguer proceeded in order to define a proof calculus, the definition of a satisfiability calculus relies on the possibility of assigning, to each theory presentation, a structure capable of expressing how its models are constructed. 

\begin{definition}[Satisfiability calculus \cite{lopezpombo:fi-166_4}]\ \\
\label{sat-calculus}
A structure $\< \mathsf{Sign}, \mathbf{Sen}, \mathbf{Mod}, \{\models^{\Sigma}\}_{\Sigma \in |\mathsf{Sign}|}, \mathbf{M}, \mathbf{Mods}, \mu \>$ is said to be a \emph{satisfiability calculus} if it satisfies the following conditions:
\begin{itemize}
\item $\< \mathsf{Sign}, \mathbf{Sen},  \mathbf{Mod},
  \{\models^{\Sigma}\}_{\Sigma \in |\mathsf{Sign}|} \>$ is an institution,
\item $\mathbf{M}: {\mathsf{Th}}^\op \to \mathsf{Struct}_{SC}$ is a functor.
 Let $T \in |\mathsf{Th}^\op|$, then $\mathbf{M}(T) \in
  |\mathsf{Struct}_{SC}|$ is the model structure of $T$,
\item $\mathbf{Mods}: \mathsf{Struct}_{SC} \to \mathsf{Cat}$ is a functor.
  Let $T \in |\mathsf{Th}^\op|$, then $\mathbf{Mods}(\mathbf{M}(T))$
  is the set of canonical models of $T$; the composite functor
  $\mathbf{Mods} \circ \mathbf{M}: {\mathsf{Th}}^\op \to \mathsf{Cat}$ will
  be denoted by $\mathbf{models}$, and
\item $\mu: \mathbf{models} \nattrans \mathscr{P} \circ \mathbf{Mod}$ is a natural transformation such that, for each $T = \<\Sigma, \Gamma\> \in |{\mathsf{Th}}^\op|$, the image of $\mu_T: \mathbf{models}(T) \to \mathscr{P} \circ \mathbf{Mod}(T)$ is the subcategory of $\mathbf{Mod}(T)$ corresponding to each canonical representation of a class of models in $|\mathbf{models}(T)|$. The map $\mu_T$ is called the \emph{projection of the category of models} of the theory $T$.
\end{itemize}
\end{definition}

The intuition behind the previous definition is the following. For any theory $T$, the functor $\mathbf{M}$ assigns a structure in the category $\mathsf{Struct}_{SC}$ representing the class of models for $T$. Notice that the target of functor $\mathbf{M}$, when applied to a theory $T$ is not necessarily a model but a structure representing the category of models of $T$. The reader may have already noticed that the functor $\mathbf{M}$ is contravariant with respect to category $\mathsf{Th}$, reflecting the existing opposite direction of morphisms between categories of models with respect to those between signatures found in institutions (see Def.~\ref{institution}). The functor $\mathbf{Mods}$ maps the structure representing the class of models of a theory $T = \<\Sigma, \Gamma\>$ to a category whose objects are canonical representations of models of $\Gamma$. Finally, for any theory $T$, the functor $\mu_T$ relates each of these structures to the corresponding subcategory of $\mathbf{Mod}(T)$. 

\begin{example}[Tableau Method for First-Order Predicate Logic]
\label{ex:sat-calculus}
\sloppy
Let us start by presenting the well-known tableaux method for first-order logic \cite{smullyan95}. Let us denote by $\mathbb{I}_{FOL} = \< \mathsf{Sign}, \mathbf{Sen}, \mathbf{Mod}, \{\models^{\Sigma}\}_{\Sigma \in |\mathsf{Sign}|} \>$ the institution of first-order predicate logic \cite[Exs.~2~and~3]{goguen:cmwlp84}. Let $\Sigma \in |\mathsf{Sign}|$ and $S \subseteq \mathbf{Sen} (\Sigma)$; then a \emph{tableau} for $S$ is a tree such that:
\begin{enumerate}
\item the nodes are labeled with sets of formulae (over $\Sigma$) and the root node is labeled with $S$,
\item if $u$ and $v$ are two connected nodes in the tree ($u$ being an ancestor of $v$), then the label of $v$ is obtained from the label of $u$ by applying one of the following rules:
\begin{center}
\begin{tabular}{lr}
\AXC{$X \cup \{A \land B\}$}
\LL{}\RL{[$\land$]}
\UnaryInfC{$X \cup \{A \land B, A, B\}$}
\DP
&
\AXC{$X \cup \{A \lor B\}$}
\LL{}\RL{[$\lor$]}
\UnaryInfC{$X \cup \{A \lor B, A\} \qquad X \cup \{A \lor B, B\}$}
\DP
\end{tabular}
\end{center}
\begin{center}
\begin{tabular}{lcr}
\AXC{$X \cup \{\neg\neg A\}$}
\LL{}\RL{[$\neg_1$]}
\UnaryInfC{$X \cup \{\neg\neg A, A\}$}
\DP
&
\AXC{$X \cup \{A\}$}
\LL{}\RL{[$\neg_2$]}
\UnaryInfC{$X \cup \{A, \neg\neg A\}$}
\DP
&
\AXC{$X \cup \{A, \neg A\}$}
\LL{}\RL{[$\mathbf{\mathit{false}}$]}
\UnaryInfC{$\mathbf{Sen} (\Sigma)$}
\DP
\end{tabular}
\end{center}
\begin{center}
\begin{tabular}{lr}
\AXC{$X \cup \{\neg(A \land B)\}$}
\LL{}\RL{[$DM_1$]}
\UnaryInfC{$X \cup \{\neg(A \land B), \neg A \lor \neg B\}$}
\DP
&
\AXC{$X \cup \{\neg(A \lor B)\}$}
\LL{}\RL{[$DM_2$]}
\UnaryInfC{$X \cup \{\neg(A \lor B), \neg A \land \neg B\}$}
\DP
\end{tabular}
\end{center}
\begin{center}
\AXC{$X \cup \{(\forall x)P(x)\}$}
\LL{[$t$ is a ground term.]}\RL{[$\forall$]}
\UnaryInfC{$X \cup \{(\forall x)P(x), P(t)\}$}
\DP
\end{center}
\begin{center}
\AXC{$X \cup \{(\exists x)P(x)\}$}
\LL{[$c$ is a new constant.]}\RL{[$\exists$]}
\UnaryInfC{$X \cup \{(\exists x)P(x), P(c)\}$}
\DP
\end{center}
\end{enumerate}
A sequence of nodes $s_0 \xrightarrow{\tau_0^{\alpha_0}} s_1 \xrightarrow{\tau_1^{\alpha_1}} s_2 \xrightarrow{\tau_2^{\alpha_2}} \ldots$ is a \emph{branch} if: 
\begin{inparaenum}[\itshape a\upshape)] 
\item $s_0$ is the root node of the tree, and 
\item for all $i \leq \omega$, $s_i \rightarrow s_{i+1}$ occurs in the tree, $\tau_i^{\alpha_i}$ is an instance of one of the rules presented above, and $\alpha_i$ are the formulae of $s_i$ to which the rule was applied. 
\end{inparaenum}
A branch $s_0 \xrightarrow{\tau_0^{\alpha_0}} s_1 \xrightarrow{\tau_1^{\alpha_1}} s_2 \xrightarrow{\tau_2^{\alpha_2}} \ldots$ in a tableau is \emph{saturated} if there exists $i \leq \omega$ such that $s_i = s_{i+1}$. A branch $s_0 \xrightarrow{\tau_0^{\alpha_0}} s_1 \xrightarrow{\tau_1^{\alpha_1}} s_2 \xrightarrow{\tau_2^{\alpha_2}} \ldots$ in a tableau is \emph{closed} if there exists $i \leq \omega$ and $\alpha \in \mathbf{Sen}(\Sigma)$ such that $\{\alpha, \neg \alpha\} \subseteq s_i$.

Let  $s_0 \xrightarrow{\tau_0^{\alpha_0}} s_1 \xrightarrow{\tau_1^{\alpha_1}} s_2 \xrightarrow{\tau_2^{\alpha_2}} \ldots$ be a branch in a tableau. Examining the rules presented above, it is straightforward to see that every $s_i$, with $i < \omega$, is a set of formulae. In each step, we have either the application of a rule decomposing one formula of the set into its constituent parts with respect to its major connective, while preserving satisfiability, or the application of the rule $[\mathit{false}]$ denoting the fact that the corresponding set of formulae is unsatisfiable. Thus, the limit set of the branch is a set of formulae containing subformulae (and ``\emph{instances}'' in the case of quantifiers) of the original set of formulae for which the tableau was built. As a result of this, every open branch represents, by means of the set of formulae occurring in the leaf, the class of models satisfying them.

In order to define the tableau method as a satisfiability calculus, we have to provide formal definitions for the categories supporting tableaux structures, for the functors $\mathbf{M}$ and $\mathbf{Mods}$ and for the natural transformation $\mu$.

First, given $\Sigma \in |\mathsf{Sign}|$ and $\Gamma \subseteq \mathbf{Sen}(\Sigma)$, we define $\mathit{Str}^{\Sigma, \Gamma} = \<\mathcal{O}, \mathcal{A}\>$ such that $\mathcal{O} = 2^{\mathbf{Sen}(\Sigma)}$ and $\mathcal{A} = \{\alpha:\{A_i\}_{i \in \mathcal{I}} \to \{B_j\}_{j \in \mathcal{J}}\ |\ \alpha=\{\alpha_j\}_{j \in \mathcal{J}}\}$, where for all $j \in \mathcal{J}$, $\alpha_j$ is a branch in a tableau for $\Gamma \cup \{B_j\}$ with leaves $\Delta \subseteq \{A_i\}_{i \in \mathcal{I}}$ $\mathit{Str}^{\Sigma, \Gamma}$ can be proved to be a category (see Lemma~\ref{lemma:str-category} for a detailed proof). 
Then, we can prove that $\<\mathit{Str}^{\Sigma, \Gamma}, \cup, \emptyset\>$, where $\cup: \mathit{Str}^{\Sigma, \Gamma} \times \mathit{Str}^{\Sigma, \Gamma} \to \mathit{Str}^{\Sigma, \Gamma}$ is the typical bi-functor on sets and functions, and $\emptyset$ is the neutral element for $\cup$, is a strict monoidal category (see Lemma~\ref{lemma:strict-monoidal-cat} for details). 

Second, using the previous definition we can introduce the class of legal tableaux (denoted by $\mathsf{Struct}_{SC}$), together with a class of arrows, and prove it is a category. $\mathsf{Struct}_{SC}$ is defined as $\<\mathcal{O}, \mathcal{A}\>^\op$ where $\mathcal{O} = \{\<\mathit{Str}^{\Sigma, \Gamma}, \cup, \emptyset\> \ |\ \Sigma \in |\mathsf{Sign}| \land \Gamma \subseteq \mathbf{Sen}(\Sigma)\}$, and $\mathcal{A} = \{\widehat{\sigma}: \<\mathit{Str}^{\Sigma, \Gamma}, \cup, \emptyset\> \to \<\mathit{Str}^{\Sigma', \Gamma'}, \cup, \emptyset\>\ |\ \sigma:\<\Sigma, \Gamma\> \to \<\Sigma', \Gamma'\> \in ||\mathsf{Th}||\}$, the homomorphic extensions of the morphisms in $||\mathsf{Th}||$ to sets of formulae preserving the application of rules (i.e., the structure of the tableaux) (see Lemma~\ref{lemma:struct-cat} for a detailed proof).

Third, the functor $\mathbf{M}$ must be understood as the relation between a theory in $|\mathsf{Th}|$ and its corresponding category of structures representing legal tableaux. So, for every theory $\<\Sigma, \Gamma\>$, $\mathbf{M}$ associates to it the strict monoidal category \cite[Sec.~1,~pp.~157]{maclane71} $\<\mathit{Str}^{\Sigma, \Gamma}, \cup, \emptyset\>$, and for every theory morphism $\sigma: \<\Sigma, \Gamma\> \to \<\Sigma', \Gamma'\> \in ||\mathsf{Th}||$ observed in the opposite direction, $\mathbf{M}$ associates to it a morphism $\widehat{\sigma}: \mathit{Str}^{\Sigma, \Gamma} \to \mathit{Str}^{\Sigma', \Gamma'}$ which is the homomorphic extension of $\sigma$ to the structure of the tableaux, also observed in the opposite direction. Then, $\mathbf{M}: \mathsf{Th}^\op \to \mathsf{Struct}_{SC}$ is defined as $\mathbf{M} (\<\Sigma, \Gamma\>) = \<\mathit{Str}^{\Sigma, \Gamma}, \cup, \emptyset\>$ and for any $\sigma: \<\Sigma, \Gamma\> \to \<\Sigma', \Gamma'\> \in ||\mathsf{Th}||$, $\mathbf{M} (\sigma^\op) = \widehat{\sigma}^\op$, where $\widehat{\sigma}: \<\mathit{Str}^{\Sigma, \Gamma}, \cup, \emptyset\> \to \<\mathit{Str}^{\Sigma', \Gamma'}, \cup, \emptyset\>$ is the homomorphic extension of $\sigma$ to the structures in $\<\mathit{Str}^{\Sigma, \Gamma}, \cup, \emptyset\>$ (see Lemma~\ref{lemma:m-functor} for a detailed proof).

Fourth, the functor $\mathbf{Mods}$ provides the means for obtaining theory presentations characterising classes of models from structures of the form $\mathit{Str}^{\Sigma, \Gamma}$ by identifying the sets of formulae in the leaves of the open branches of a tableau. To this effect, $\mathbf{Mods}: {\mathsf{Struct}_{SC}} \to \mathsf{Cat}$ is defined on objects as $\mathbf{Mods} (\<\mathit{Str}^{\Sigma, \Gamma}, \cup, \emptyset\>) = \<\mathcal{O}, \mathcal{A}\>$ where: 
\begin{itemize}
\item $\mathcal{O} = \bigcup_{\< \Sigma, \Delta \> \in |\mathsf{Th}|} \setof{\<\Sigma, \widetilde{\Delta}\> \in |\mathsf{Th}|}{
\begin{array}{l}
(\exists \alpha: \Delta \to \emptyset \in ||\mathit{Str}^{\Sigma, \Gamma}||)\\
(\forall \alpha': \Delta' \to \Delta \in ||\mathit{Str}^{\Sigma, \Gamma}||)\\
\qquad (\Delta' = \Delta) \land (\widetilde{\Delta} \to \emptyset \in \alpha) \land\\
\qquad \neg(\exists \varphi)(\{\neg\varphi, \varphi\} \subseteq \widetilde{\Delta}) \}
\end{array}}$,\\
\noindent (i.e., 
\begin{inparaenum}[1.] 
\item the existentially quantified $\alpha: \Delta \to \emptyset \in ||\mathit{Str}^{\Sigma, \Gamma}||$ is a tableau, 
\item stating that for any $\alpha': \Delta' \to \Delta \in ||\mathit{Str}^{\Sigma, \Gamma}||$, the equation $\Delta' = \Delta$ holds, expresses that $\alpha$ is saturated, and
\item requesting that $\widetilde{\Delta}$, where $\widetilde{\Delta} \to \emptyset$ is a branch of $\alpha$, is not inconsistent, expresses that it is the set of formulae at the leaf of an open branch\end{inparaenum}), and
\item $\mathcal{A} = \{ id_{T}: T \to T\ |\ T \in \mathcal{O} \}$, (i.e., only the identities);
\end{itemize}
\noindent and on morphisms as for all $\sigma: \<\Sigma, \Gamma\> \to \<\Sigma', \Gamma'\> \in ||\mathsf{Th}||$, $\mathbf{Mods}(\widehat{\sigma}^\op)(\<\Sigma, \delta\>) = \<\Sigma', \mathbf{Sen}(\sigma)(\delta)\>$. This is proved to be a functor (see Lemma~\ref{lemma:mods-functor} for a detailed proof).

Finally, $\mu$ has to relate the structures representing saturated tableaux with the model satisfying the set of formulae denoted by the source of the morphism; then we can  define $\mu_{\<\Sigma, \Delta\>}: \mathbf{models} (\<\Sigma, \Delta\>) \to \mathscr{P} \circ \mathbf{Mod} (\<\Sigma, \Delta\>)$ as for all $\<\Sigma, \delta\> \in |\mathbf{models}(\<\Sigma, \Delta\>)|$, $\mu_{\<\Sigma, \Delta\>} (\<\Sigma, \delta\>) = \mathbf{Mod} (\<\Sigma, \delta\>)$ and prove it to be a natural transformation (see Lemma~\ref{lemma:mu-nattrans} for details).

Therefore, we conclude that $\< \mathsf{Sign}, \mathbf{Sen}, \mathbf{Mod}, \{\models^{\Sigma}\}_{\Sigma \in |\mathsf{Sign}|}, \mathbf{M}, \mathbf{Mods}, \mu \>$ is a satisfiability calculus.

\begin{definition}
Let $\Sigma \in |\mathsf{Sign}|$ and $\Gamma \subseteq \mathbf{Sen}(\Sigma)$, then we define $\mathit{Str}^{\Sigma, \Gamma} = \<\mathcal{O}, \mathcal{A}\>$ such that $\mathcal{O} = 2^{\mathbf{Sen}(\Sigma)}$ and $\mathcal{A} = \{\alpha:\{A_i\}_{i \in \mathcal{I}} \to \{B_j\}_{j \in \mathcal{J}}\ |\ \alpha=\{\alpha_j\}_{j \in \mathcal{J}}\}$, where for all $j \in \mathcal{J}$, $\alpha_j$ is a branch in a tableau for $\Gamma \cup \{B_j\}$ with leaves $\Delta \subseteq \{A_i\}_{i \in \mathcal{I}}$. It should be noted that $\Delta \models^\Sigma \Gamma \cup \{B_j\}$. 
\end{definition}

\begin{lemma}
Let $\Sigma \in |\mathsf{Sign}|$ and $\Gamma \subseteq \mathbf{Sen}(\Sigma)$; then $\<\mathit{Str}^{\Sigma, \Gamma}, \cup, \emptyset\>$, where $\cup: \mathit{Str}^{\Sigma, \Gamma} \times \mathit{Str}^{\Sigma, \Gamma} \to \mathit{Str}^{\Sigma, \Gamma}$ is the typical bi-functor on sets and functions, and $\emptyset$ is the neutral element for $\cup$, is a strict monoidal category. 
\end{lemma}
\begin{proof}
The proof can be found in Lemma~\ref{lemma:strict-monoidal-cat}.
\end{proof}

Using this definition we can introduce the category of legal tableaux, denoted by $\mathsf{Struct}_{SC}$.

\begin{definition}
\label{def-struct}
$\mathsf{Struct}_{SC}$ is defined as $\<\mathcal{O}, \mathcal{A}\>$ where $\mathcal{O} = \{\mathit{Str}^{\Sigma, \Gamma}\ |\ \Sigma \in |\mathsf{Sign}| \land \Gamma \subseteq \mathbf{Sen}(\Sigma)\}$, and $\mathcal{A} = \{\widehat{\sigma}: \mathit{Str}^{\Sigma, \Gamma} \to \mathit{Str}^{\Sigma', \Gamma'}\ |\ \sigma:\<\Sigma, \Gamma\> \to \<\Sigma', \Gamma'\> \in ||\mathsf{Th}||\}$, the homomorphic extension of the morphisms in $||\mathsf{Th}||$.
\end{definition}

\begin{lemma}
$\mathsf{Struct}_{SC}$ is a category.
\end{lemma}
\begin{proof}
The proof can be found in Lemma~\ref{lemma:struct-cat}.
\end{proof}

The functor $\mathbf{M}$ must be understood as the relation between a theory in $|\mathsf{Th}|$ and its category of structures representing legal tableaux. So, for every theory $T$, $\mathbf{M}$ associates the strict monoidal category \cite{maclane71} $\<\mathit{Str}^{\Sigma, \Gamma}, \cup, \emptyset\>$, and for every theory morphism $\sigma: \<\Sigma, \Gamma\> \to \<\Sigma', \Gamma'\>$, $\mathbf{M}$ associates a morphism $\widehat{\sigma}: \mathit{Str}^{\Sigma, \Gamma} \to \mathit{Str}^{\Sigma', \Gamma'}$ which is the homomorphic extension of $\sigma$ to the structure of the tableaux.

\begin{definition}
\label{def-m}
$\mathbf{M}: \mathsf{Th}^\op \to \mathsf{Struct}_{SC}$ is defined as $\mathbf{M} (\<\Sigma, \Gamma\>) = \<\mathit{Str}^{\Sigma, \Gamma}, \cup, \emptyset\>$ and $\mathbf{M} (\sigma: \<\Sigma, \Gamma\> \to \<\Sigma', \Gamma'\>) = \widehat{\sigma}: \<\mathit{Str}^{\Sigma, \Gamma}, \cup, \emptyset\> \to \<\mathit{Str}^{\Sigma', \Gamma'}, \cup, \emptyset\>$, the homomorphic extension of $\sigma$ to the structures in $\<\mathit{Str}^{\Sigma, \Gamma}, \cup, \emptyset\>$. 
\end{definition}

\begin{lemma}
\label{m-functor}
$\mathbf{M}$ is a functor.
\end{lemma}
\begin{proof}
The proof can be found in Lemma~\ref{lemma:m-functor}.
\end{proof}

In order to define $\mathbf{Mods}$, we need the following auxiliary definition, which provides the maximal consistent sets of formulae for a given set of formulae.

\begin{definition}
Let  $\Sigma \in |\mathsf{Sign}|$, $\Delta \subseteq \mathbf{Sen}(\Sigma)$, then $Cn (\Delta)$ is defined as follows:
$$Cn (\Delta) = \{ \overline{\Delta}\ |\ (\forall \alpha \in \mathbf{Sen}(\Sigma))(\alpha \not\in \overline{\Delta} \implies |\mathbf{Mod}(\{\alpha\} \cup \overline{\Delta})| = \emptyset)\}$$
\end{definition}

Given $\<\Sigma, \Gamma\> \in |\mathsf{Th}|$, the functor $\mathbf{Mods}$ provides the means for obtaining the category containing all the possible maximal consistent extensions of those structures in $\mathit{Str}^{\Sigma, \Gamma}$ representing the closure of the branches in saturated tableaux. 

\begin{definition}
\label{def-mods}
$\mathbf{Mods}: {\mathsf{Struct}_{SC}} \to \mathsf{Cat}$ is defined as 
\begin{itemize}
\item $\mathbf{Mods} (\<\mathit{Str}^{\Sigma, \Gamma}, \cup, \emptyset\>) = \bigcup\{Cn(\widetilde{\Delta})\ |\ (\exists \alpha: \Delta \to \emptyset \in |\mathit{Str}^{\Sigma, \Gamma}|) (\widetilde{\Delta} \to \emptyset \in \alpha \land (\forall \alpha': \Delta' \to \Delta \in ||\mathit{Str}^{\Sigma, \Gamma}||)(\Delta' = \Delta))\}$, and
\item forall $\sigma: \<\Sigma, \Gamma\> \to \<\Sigma', \Gamma'\>$ such that $\mathbf{M}(\sigma) = \widehat{\sigma}$, $\mathbf{Mods}(\widehat{\sigma})(\Delta) = \mathbf{Sen}(\sigma)(\Delta)$. 
\end{itemize}
\end{definition}

\begin{lemma}
\label{mods-functor}
$\mathbf{Mods}$ is a functor.
\end{lemma}
\begin{proof}
The proof can be found in Lemma~\ref{lemma:mods-functor}.
\end{proof}

Finally, the natural transformation $\mu$ relates the structures representing saturated tableaux with the model satisfying the set of formulae denoted by the source of the morphism.

\begin{definition}
\label{def-mu}
Let $T \in |\mathsf{Th}|$, then we define $\mu_T: \mathbf{models} (T) \to \mathsf{Cat}$ as $\mu_T (\Delta) = \mathbf{Mod}(\<\mathbf{Sign}(\Delta), \Delta\>)$\footnote{$\mathbf{Sign}:\mathsf{Th} \to \mathsf{Sign}$ is the forgetful functor that projects the signature from theory presentations.}.
\end{definition}

\begin{fact}
Let $\<\Sigma, \Gamma\> \in |\mathsf{Th}|$, then we define $\mu_{\Sigma}: \mathbf{models}^\op (\<\Sigma, \Gamma\>) \to \mathbf{Mod} (\<\Sigma, \Gamma\>)$ defined as in Def.~\ref{def-mu}. Let $\Sigma \in |\mathsf{Sign}|$ and $\Gamma \subseteq \mathbf{Sen} (\Sigma)$, then $\mu_{\<\Sigma, \Gamma\>}$ is a functor.
\end{fact}

\begin{lemma}
\label{mu-nattrans}
$\mu$ is a natural transformation.
\end{lemma}
\begin{proof}
The proof can be found in Lemma~\ref{lemma:mu-nattrans}.
\end{proof}

Now, from Lemmas~\ref{m-functor},~\ref{mods-functor},~and~\ref{mu-nattrans}, and considering the hypothesis that $\mathbb{I}_{FOL}$ is an institution, the following corollary follows.

\begin{corollary}
$\< \mathsf{Sign}, \mathbf{Sen}, \mathbf{Mod}, \{\models^{\Sigma}\}_{\Sigma \in |\mathsf{Sign}|}, \mathbf{M}, \mathbf{Mods}, \mu \>$ is a satisfiability calculus.
\end{corollary}

The reader interested in the details of the proofs is pointed to \cite[Ex.~3.2]{lopezpombo:fi-166_4}.
\end{example}

As we mentioned in the introduction, model theory and proof systems are closely related and, in general, logical properties such as soundness and completeness are established by relating the families of entailment and satisfaction relations as we showed in Def.~\ref{logic}. As a consequence of having refined the notions of entailment and satisfaction by introducing concrete representations (i.e. proof calculus and satisfiability calculus, respectively), soundness and completeness can be reformulated in terms of the structures used for representing proofs and models.

The following proposition gives an alternative definition of soundness and completeness.

\begin{proposition}
\label{soundness-completeness}
Let the structure $\< \mathsf{Sign}, \mathbf{Sen}, \mathbf{Mod}, \{\vdash^{\Sigma}\}_{\Sigma \in |\mathsf{Sign}|}, \{\models^{\Sigma}\}_{\Sigma \in |\mathsf{Sign}|} \>$ be a logic, the structure $\< \mathsf{Sign}, \mathbf{Sen}, \{\vdash^{\Sigma}\}_{\Sigma \in |\mathsf{Sign}|}, \mathbf{P}, \mathbf{Pr}, \pi \>$ be a proof calculus and the structure $\< \mathsf{Sign}, \mathbf{Sen}, \mathbf{Mod}, \{\vdash^{\Sigma}\}_{\Sigma \in |\mathsf{Sign}|}, \{\models^{\Sigma}\}_{\Sigma \in |\mathsf{Sign}|}, \mathbf{M}, \mathbf{Mods}, \mu \>$ be a satisfiability calculus. Therefore, if $T = \< \Sigma, \Gamma \> \in |\mathsf{Th}_0|$ and $\alpha \in \mathbf{Sen}(\Sigma)$:\\

\noindent \emph{\textbf{[Soundness]}} If there exists $\tau \in |\mathbf{proof}(T)|$ such that $\pi_T(\tau) = \alpha$, then for all $M \in |\mathbf{models} (T)|$ and $\mathcal{M} \in |\mu_T (M)|$, $\mathcal{M} \models^\Sigma \alpha$.\\

\noindent \emph{\textbf{[Completeness]}} If for all $M \in |\mathbf{models}(T)|$ and $\mathcal{M} \in |\mu_T (M)|$, $\mathcal{M} \models^\Sigma \alpha$, then there exists $\tau \in |\mathbf{proof}(T)|$ such that $\pi_T(\tau) = \alpha$.
\end{proposition}

Therefore, soundness can be reformulated in order to support the known fact stating that the evidence of the existence of a counterexamples is a sufficient fact to conclude that there cannot exists a proof.

\begin{corollary}
\label{pre-coro-soundness}
Let the structure $\< \mathsf{Sign}, \mathbf{Sen}, \mathbf{Mod}, \{\vdash^{\Sigma}\}_{\Sigma \in |\mathsf{Sign}|}, \{\models^{\Sigma}\}_{\Sigma \in |\mathsf{Sign}|} \>$ be a logic, the structure $\< \mathsf{Sign}, \mathbf{Sen}, \{\vdash^{\Sigma}\}_{\Sigma \in |\mathsf{Sign}|}, \mathbf{P}, \mathbf{Pr}, \pi \>$ be a proof calculus and the structure $\< \mathsf{Sign}, \mathbf{Sen}, \mathbf{Mod}, \{\models^{\Sigma}\}_{\Sigma \in |\mathsf{Sign}|}, \mathbf{M}, \mathbf{Mods}, \mu \>$ be a satisfiability calculus. Therefore, for all $T = \< \Sigma, \Gamma \> \in |\mathsf{Th}_0|$ and $\alpha \in \mathbf{Sen}(\Sigma)$, if there exists $M \in |\mathbf{models}(T)|$ and $\mathcal{M} \in |\mu_T (M)|$ such that $\mathcal{M} \models^\Sigma \alpha$ does not hold, then there is no $\tau \in |\mathbf{proof}(T)|$ such that $\pi_T(\tau) = \alpha$.
\end{corollary}
 
Now, on the one hand, Def.~\ref{sat-calculus} provides the means for defining specific mechanisations of the satisfiability relation $\models_\Sigma$, with $\Sigma \in |\mathsf{Sign}|$, but, on the other hand, being able to use Coro.~\ref{coro-soundness} requires asserting that a model $\mathcal{M}$ does not satisfies a formula $\alpha$, written above as ``$\mathcal{M} \models_\Sigma \alpha$ does not hold''. Luckily, most of the logical languages have some form of negation in their syntax that, from a semantical point of view, translates into a complementation of the notion of satisfaction. This notion is usually referred to as for an institution to ``have negation'', understood as, for the language of that institution, to have a syntactic device whose behaviour, from a model-theoretic perspective, is that of a negation. The next definition formalises this notion.

\begin{definition}
\label{have-negation}
Given a institution $\mathbb{I}$, we say that \emph{$\mathbb{I}$ has negation} if for every signature $\Sigma \in |\mathsf{Sign}|$, $\alpha \in \mathbf{Sen} (\Sigma)$, there exists a formula in $\mathbf{Sen} (\Sigma)$, usually denoted as $\neg \alpha$, such that for all $\mathcal{M} \in \mathbf{Mod} (\Sigma)$,
$$\mathcal{M} \models^\Sigma \neg\alpha \text{ iff } \mathcal{M} \models^\Sigma \alpha \text{ does not hold.}$$
\end{definition}

Therefore, Coro.~\ref{pre-coro-soundness} can be reformulated as follows.

\begin{corollary}
\label{coro-soundness}
Let the structure $\< \mathsf{Sign}, \mathbf{Sen}, \mathbf{Mod}, \{\vdash^{\Sigma}\}_{\Sigma \in |\mathsf{Sign}|}, \{\models^{\Sigma}\}_{\Sigma \in |\mathsf{Sign}|} \>$ be a logic, the structure $\< \mathsf{Sign}, \mathbf{Sen}, \{\vdash^{\Sigma}\}_{\Sigma \in |\mathsf{Sign}|}, \mathbf{P}, \mathbf{Pr}, \pi \>$ be a proof calculus and the structure $\< \mathsf{Sign}, \mathbf{Sen}, \mathbf{Mod}, \{\models^{\Sigma}\}_{\Sigma \in |\mathsf{Sign}|}, \mathbf{M}, \mathbf{Mods}, \mu \>$ be a satisfiability calculus such that $\< \mathsf{Sign}, \mathbf{Sen}, \mathbf{Mod}, \{\models^{\Sigma}\}_{\Sigma \in |\mathsf{Sign}|}\>$ has negation. Therefore, for all $T = \< \Sigma, \Gamma \> \in |\mathsf{Th}_0|$ and $\alpha \in \mathbf{Sen}(\Sigma)$, if there exists $M \in |\mathbf{models}(T)|$ and $\mathcal{M} \in |\mu_T (M)|$ such that $\mathcal{M} \models^\Sigma \neg\alpha$, then there is no $\tau \in |\mathbf{proof}(T)|$ such that $\pi_T(\tau) = \alpha$.
\end{corollary}


\section{Satisfiability subcalculus, effectiveness and the foundations for combining deduction and model finding}
\label{subsat}
In the introduction we mentioned the tool Dynamite \cite{frias:tacas07} as the result from combining the capabilities of the bounded model checker Alloy \cite{jackson:acmtosem-11_2} and a theorem prover for the language of first-order relational logic \cite{jackson:acm-sigsoft00}, implemented on top of the theorem prover PVS. In this section we present the formal elements underlying this way of combining (bounded) counterexample finding with theorem proving from a language independent setting.\\

As we mentioned in Sec.~\ref{intro}, in \cite{meseguer:lc87}, Meseguer not only developed the idea of formalising the notion of proof calculus as an ``implementation'' of an entailment system, but also explored the possibility of restricting a proof calculus in such a way that the result is a subcalculus for a restriction of the language with specific properties. The next definition, introduced in \cite{lopezpombo:fi-166_4} presentes a model-theoretic counterpart of Meseguer's definitions of proof subcalculus.

\begin{definition}[Satisfiability subcalculus]
\label{sat-subcalculus}
\sloppy
A \emph{satisfiability subcalculus} is a structure of the form $\< \mathsf{Sign}, \mathbf{Sen}, \mathbf{Mod},  \mathsf{Sign}_0, \mathbf{ax},  \{\models^{\Sigma}\}_{\Sigma \in |\mathsf{Sign}|}, \mathbf{M}, \mathbf{Mods}, \mu \>$ satisfying the following conditions: 
\begin{itemize}
\item $\< \mathsf{Sign}, \mathbf{Sen},  \mathbf{Mod}, \{\models^{\Sigma}\}_{\Sigma \in |\mathsf{Sign}|} \>$ is an institution,
\item $\mathsf{Sign}_0$ is a subcategory of $ \mathsf{Sign}$ called the subcategory of admissible signatures; the restriction of the functor $\mathbf{Sen}$ to  $\mathsf{Sign_0}$ will be denoted by  $\mathbf{Sen}_0$,
\item $\mathbf{ax}: \mathsf{Sign}_0 \to \mathsf{Set}$ is a subfunctor of the functor obtained by composing $\mathbf{Sen}_0$ with the powerset functor, i.e., there is a natural inclusion $\mathbf{ax}(\Sigma) \subseteq \mathscr{P}(\mathbf{Sen}_0(\Sigma))$ for each
$\Sigma \in \mathsf{Sign}_0$. Each $\Gamma \in \mathbf{ax}(\Sigma)$ is called a set of admissible axioms.
This defines a subcategory $\mathsf{Th}_{ax}$ of $\mathsf{Th}$ whose objects are theory presentations $T = \<\Sigma, \Gamma\>$ with
$\Sigma \in \mathsf{Sign}_0$ and $\Gamma \in \mathbf{ax}(\Sigma)$, and whose morphisms are axiom-preserving theory
morphisms $H$ such that $H$ is in $\mathsf{Sign}_0$.
\item $\mathbf{M}: {\mathsf{Th}_{ax}}^\op \to \mathsf{Struct}_{SC}$ is a functor.
 Let $T \in |{\mathsf{Th}_{ax}}^\op|$, then $\mathbf{M}(T) \in
  |\mathsf{Struct}_{SC}|$ is the model structure of $T$,
\item $\mathbf{Mods}: \mathsf{Struct}_{SC} \to \mathsf{Cat}$ is a functor.
  Let $T \in |{\mathsf{Th}_{ax}}^\op|$, then $\mathbf{Mods}(\mathbf{M}(T))$
  is the set of canonical models of $T$; the composite functor
  $\mathbf{Mods} \circ \mathbf{M}: {\mathsf{Th}_{ax}}^\op \to \mathsf{Cat}$ will
  be denoted by $\mathbf{models}$, and
  \item $\mu: \mathbf{models} \nattrans \mathscr{P} \circ \mathbf{Mod}$ is a natural transformation such that, for each $T = \<\Sigma, \Gamma\> \in |{\mathsf{Th}_{ax}}^\op|$, the image of $\mu_T: \mathbf{models}(T) \to \mathscr{P} \circ \mathbf{Mod}(T)$ is the subcategory of $\mathbf{Mod}(T)$ corresponding to each canonical representation of a class of models in $|\mathbf{models}(T)|$. The map $\mu_T$ is called the \emph{projection of the category of models} of the theory $T$.
\end{itemize}
\end{definition}

There are no major differences with respect to \cite[Def.~14]{meseguer:lc87}, except for the lack of restriction on the possible conclusions that the sub-calculus can draw. Notice that in the case of proof theoretic approaches, the introduction of a functor restricting the conclusions, as a sub functor of $\mathbf{Sen}$, is a key element of the definition introduced for restricting the set of theorems of interest, in contrast with satisfiability calculi, where models only relate to a given theory presentation, there is not such need.\\

\begin{example}[Tableau Method as a satisfiability subcalculus for the quantifier-free and ground fragment of first-order predicate logic]
\label{ex:sat-subcalculus}
This example strongly relays on Ex.~\ref{ex:sat-calculus} as the language for which we present a satisfiability subcalculus is a restriction of the language of first-order predicate logic.

We first proceed to define the basic restriction of the language of first order predicate logic by determining $\mathsf{Sign}_0$ and $\mathbf{ax}$:
\begin{itemize}
 \item $\mathsf{Sign}_0 = \mathsf{Sign}$,
 \item $\mathbf{ax}$ is the composition of a functor $\mathit{Qf}: \mathsf{Sign}_0 \to \mathsf{Set}$ yielding, for each $\Sigma \in |\mathsf{Sign}|$, the restriction of $\mathbf{Sen}(\Sigma)$ to the quantifier-free sentences, with $\mathscr{P}_{\mathit{fin}}$.
\end{itemize}

Once we established the restriction of the language, $\mathbf{M}$ is defined exactly as in Ex.~\ref{ex:sat-calculus} but considering only those tableaux that do not contain any application of rules [$\forall$] and [$\exists$]. Notice that these rules are not required as they can only be applied to sets of formulae containing a quantified formula which, by the restriction stated before, do not exists.

Finally, $\mathbf{Mods}$ and $\mu$ are defined exactly as in Ex.~\ref{ex:sat-calculus}.
\end{example}

The next definition introduces the notion of effectiveness of the procedure for constructing structures characterising models for sets of sentences over a logical (sub)language. To this end, we follow Meseguer’s approach \cite[Def.~16]{meseguer:lc87}, which, in turn, adopts the axiomatic view of computability outlined by Shoenfield in \cite{schoenfield71}. The elementary notions are those of a \emph{finite object}, a \emph{space} of finite objects, and \emph{recursive functions}. In Shoenfield's own words, a \emph{finite object} is an ``object which can be specified by a finite amount of information'', a \emph{space} is ``an infinite class $X$ of finite objects such that, given a finite object $x$, we can decide whether or not $x$ belongs to $X$''. Now, given spaces $X$ and $Y$, a recursive function $f : X \to Y$ is then a total function that can be computed by an algorithm (i.e., by a terminating program, disregarding space and time limitations); spaces and recursive functions form a category $\mathsf{Space}$. Effectiveness is then obtained by restricting sentences and axioms over selected signatures to be organised in a space. 

\begin{definition}[Effective satisfiability subcalculus]
\label{eff-sat-subcalculus}
\sloppy
An \emph{effective satisfiability subcalculus} is a structure $\mathbbm{Q}$ of the form:
$\< \mathsf{Sign}, \mathbf{Sen}, \mathbf{Mod},  \mathsf{Sign}_0, \mathbf{Sen}_0, \mathbf{ax},  \{\models^{\Sigma}\}_{\Sigma \in |\mathsf{Sign}|}, \mathbf{M}, \mathbf{Mods}, \mu \>$ satisfying the following conditions:
\begin{itemize}
\item $\< \mathsf{Sign}, \mathbf{Sen},  \mathbf{Mod}, \{\models^{\Sigma}\}_{\Sigma \in |\mathsf{Sign}|} \>$ is an institution.
\item $\mathsf{Sign}_0$ is a subcategory of $ \mathsf{Sign}$ called the subcategory of admissible signatures; let $J: \mathsf{Sign}_0 \hookrightarrow \mathsf{Sign}$ be the inclusion functor.
\item $\mathbf{Sen}_0:\mathsf{Sign}_0 \to \mathsf{Space}$ is a functor such that $\mathcal{U} \circ \mathbf{Sen}_0 = \mathbf{Sen} \circ J$, where $\mathcal{U}: \mathsf{Space} \to \mathsf{Set}$ is the obvious forgetful functor.
\item $\mathbf{ax}: \mathsf{Sign}_0 \to \mathsf{Space}$ is a subfunctor of the functor obtained by composing $\mathbf{Sen}_0$ with the functor $\mathscr{P}_{\mathit{fin}} : \mathsf{Space} \to \mathsf{Space}$, that sends each space to the space of its finite subsets. This defines a subcategory $\mathsf{Th}_{ax}$ of $\mathsf{Th}$ whose objects are theories $T = \<\Sigma, \Gamma\>$ with $\Sigma \in \mathsf{Sign}_0$ and $\Gamma \in \mathbf{ax}(\Sigma)$, and whose morphisms are axiom-preserving theory morphisms $H$ such that $H$ is in $\mathsf{Sign}_0$.
\item $\mathbf{M}: {\mathsf{Th}_{ax}}^\op \to \mathsf{Struct}_{SC}$ is a functor.
 Let $T \in |{\mathsf{Th}_{ax}}^\op|$, then $\mathbf{M}(T) \in
  |\mathsf{Struct}_{SC}|$ is the model structure of $T$.
\item $\mathbf{Mods}: \mathsf{Struct}_{SC} \to \mathsf{Space}$ is a functor.
  Let $T \in |{\mathsf{Th}_{ax}}^\op|$, then $\mathbf{Mods}(\mathbf{M}(T))$
  is the set of canonical models of $T$; the composite functor
  $\mathbf{Mods} \circ \mathbf{M}: {\mathsf{Th}_{ax}}^\op \to \mathsf{Space}$ will
  be denoted by $\mathbf{models}$.
\item $\mu: \mathbf{models} \nattrans \mathscr{P} \circ \mathbf{Mod}$ is a natural transformation such that, for each $T = \<\Sigma, \Gamma\> \in |{\mathsf{Th}_{ax}}^\op|$, the image of $\mu_T: \mathbf{models}(T) \to \mathscr{P} \circ \mathbf{Mod}(T)$ is the subcategory of $\mathbf{Mod}(T)$ corresponding to each canonical representation of a class of models in $|\mathbf{models}(T)|$. The map $\mu_T$ is called the \emph{projection of the category of models} of the theory $T$.
\item $\mathcal{U} (\mathbb{Q}) = \< \mathsf{Sign}, \mathbf{Sen}, \mathbf{Mod},  \mathsf{Sign}_0, \mathcal{U} \circ \mathbf{ax},  \{\models^{\Sigma}\}_{\Sigma \in |\mathsf{Sign}|}, \mathbf{M}, \mathcal{U} \circ \mathbf{Mods}, \mu \circ \mathcal{U}\>$ is a satisfiability subcalculus, where $\mu \circ \mathcal{U}$ denotes the natural transformation formed by $\{\mu_T \circ \mathcal{U}: \mathcal{U} \circ \mathbf{models} (T) \to \mathscr{P} \circ \mathbf{Mod}\}_{T \in |{\mathsf{Th}_{\mathit{ax}}}^\op|}$.
\end{itemize}
\end{definition}

Notice that the only differences between Def.~\ref{sat-subcalculus} and Def.~\ref{eff-sat-subcalculus} are that: 
\begin{inparaenum}[a)]  
\item \label{rest} the restriction $\mathsf{Sen}_0$ of the functor $\mathsf{Sen}$ must satisfy $\mathcal{U} \circ \mathsf{Sen}_0 = \mathsf{Sen} \circ J$, and
\item \label{des} all occurrences of the category $\mathsf{Set}$ involved in the definition are replaced by $\mathsf{Space}$.
\end{inparaenum}
Condition~\ref{rest} establishes that the restriction $\mathsf{Sen}_0: \mathsf{Sign}_0 \to \mathsf{Space}$ of the functor $\mathsf{Sen}: \mathsf{Sign} \to \mathsf{Set}$ is coherent with the inclusion functor $J$ and the forgetful functor $\mathcal{U}$, and Cond.~\ref{des} ensures decidability.

\begin{example}[Truth tables as an effective satisfiability subcalculus for the quantifier-free and ground fragment of first-order predicate logic]
\label{ex:eff-sat-subcalculus}
Let $\mathcal{X}$ be a set of first-order variables and $\mathbb{I}_{\text{FOL}}$ be the institution for first-order predicate logic. Let $\Sigma$ be a set of propositional variables then, the set of propositional sentences over those variables (denoted as $\mathcal{L}_{\Sigma}$) is defined as follows:
\begin{center}
 $\mathcal{L}_{\Sigma} := p \mid \neg \alpha \mid \alpha \lor \beta$ where $p \in \Sigma$ and $\alpha, \beta \in \mathcal{L}_{\Sigma}$
\end{center}
\noindent an \emph{assignment} for $\Sigma$ is a function $\mathit{val}_{\Sigma}: \Sigma \to \{ \bot, \top \}$. Given $\alpha \in \mathcal{L}_{\Sigma}$, we say that an \emph{assignment} $\mathit{val}$ satisfies $\alpha$ (denoted as $\mathit{val} \models \alpha$) if and only if:
$$
\begin{array}{rcl}
\mathit{val} \models p & \text{iff} & \mathit{val} (p) = \top\\
\mathit{val} \models \neg \alpha & \text{iff} & \mathit{val} \models \alpha \text{ does not hold}\\
\mathit{val} \models \alpha \lor \beta & \text{iff} & \mathit{val} \models \alpha \text{ or } \mathit{val} \models \beta
\end{array}
$$

Given a set of formulae $\Gamma \in 2^{\mathcal{L}_\Sigma}$ we say that a formula $\alpha$ \emph{can be obtained from} $\Gamma$ if and only if at least one of the following conditions hold:
\begin{itemize}
 \item $\alpha \in \Gamma$
 \item $\alpha$ is of the form $\neg \beta_1$
 \item $\alpha$ is of the form $\beta_1 \lor \beta_2$
\end{itemize}
\noindent where $\beta_1$ and $\beta_2$ can be obtained from $\Gamma$.

The reader should note that, if we define $\Gamma$ as a set of propositional variables then a formula $\alpha$ \emph{can be obtained from} $\Gamma$ if and only if $\alpha \in \mathcal{L}_{\Gamma}$; furthermore, we can tell if a \emph{truth assignment} $\mathit{val}_{\Gamma}$ satisfies the formula $\alpha$ (noted $\mathit{val}_{\Gamma} \models \alpha$). 

Given $\Gamma_1, \Gamma_2 \in \mathcal{L}_\Sigma$, if all formula $\alpha \in \Gamma_2$ can be obtained from $\Gamma_1$, then we can define a \emph{Truth table} from $\Gamma_1$ to $\Gamma_2$ as a function $t : \mathit{assigns}_{\Gamma_1} \to \mathit{assigns}_{\Gamma_2}$. Each set $\mathit{assigns}_{\Gamma}$ is the set of all functions with type $\Gamma \to \{\top, \bot\}$; in particular, the set $\mathit{assigns}_{\Sigma}$ is the set of all possible \emph{truth assignments} for $\Sigma$. In order for $t$ to be a \emph{Truth table} the following condition must also hold:
\begin{center}
$t(\mathit{val}_{\Gamma_1})(\alpha) = \top$ if and only if $\mathit{val}_{\Gamma_1} \models \alpha$, for all $\mathit{val}_{\Gamma_1} \in \mathit{assigns}_{\Gamma_1}$, for $\alpha \in \Gamma_2$.
\end{center}


Note that given two \emph{Truth tables} $t : \mathit{assigns}_{\Gamma_1} \to \mathit{assigns}_{\Gamma_2}$ and $t' : \mathit{assigns}_{\Gamma_2} \to \mathit{assigns}_{\Gamma_3}$, the composition $t' \circ t : \mathit{assigns}_{\Gamma_1} \to \mathit{assigns}_{\Gamma_3}$ is also a \emph{Truth table}.
Next, we have to provide formal definitions for the categories supporting \emph{Truth table} structures. Given a set of propositional variables $\Sigma$ and a set of formulae $\Gamma \in 2^{\mathcal{L}_\Sigma}$, we define $\mathit{Str}^{\Sigma, \Gamma} = \< \mathcal{O}, \mathcal{A} \>$ such that $\mathcal{O} = 2^{\Gamma'}$ where $\Gamma'$ is the set of all subformulae of the formulae in $\Gamma$, and $\mathcal{A} = \{t : \Gamma_1 \to \Gamma_2 \mid \Gamma_1, \Gamma_2 \subseteq 2^{\Gamma'} \}$, where all formula $\alpha \in \Gamma_2$ can be obtained from $\Gamma_1$ and $t$ is the \emph{Truth table} from $\Gamma_1$ to $\Gamma_2$. Note that a structure $\mathit{Str}^{\Sigma, \Gamma}$ has identities and composition, and it can be proved to be a category.

Now we can define the effective satisfiability subcalculus for $\mathbbm{I}_\mathit{FOL}$ as follows:
\begin{enumerate}
 \item In this case any signature is an admissible signature, therefore  $\mathsf{Sign}_0 = \mathsf{Sign}$. Then, $\mathbf{Sen}_0$ is the restriction of the functor $\mathbf{Sen}$ to the category $\mathsf{Sign}_0$. Note that $\mathbf{Sen}_0(\Sigma)$ is a recursive infinite set of finite objects, and therefore, it can be proved that the image of $\mathbf{Sen}_0$ is a space. Since $\mathbf{Sen}_0$ is a subfunctor of $\mathbf{Sen}$ it is easy to observe that $\mathcal{U} \circ \mathbf{Sen}_0 = \mathbf{Sen} \circ J$.
 \item $\mathbf{ax}$ is the composition of a functor $\mathit{Qf}: \mathsf{Sign}_0 \to \mathsf{Set}$ yielding, for each $\Sigma \in |\mathsf{Sign}|$, the restriction of $\mathbf{Sen}_0(\Sigma)$ to the quantifier-free sentences, with $\mathscr{P}_{\mathit{fin}}$.

 \item In order to define $\mathbf{M}$ we first need to map each first order formula to a boolean combination of propositional variables and then introduce the class of legal Truth tables for those formulas. We will call $\mathsf{prop}(\Sigma) = \{v_\mathit{p} \mid \mathit{p}$ is a formula in $\Gamma$, with $\Gamma \in \mathbf{ax}(\Sigma) \}$ for $\Sigma \in \mathsf{Sign}_0$. Then, we define the function $\mathit{Tr}: \mathit{Qf}(\Sigma) \to \mathcal{L}_{\mathsf{prop}(\Sigma)}$ as follows:
 $$
\begin{array}{rcl}
\mathit{Tr} (P (t_1, \dots , t_k)) & = & v_{\mbox{``$P(t_1, \ldots, t_k)$''}} \\ 
\mathit{Tr} (t = t') & = & v_{\mbox{``$t = t'$''}},\\
\mathit{Tr} (\neg \alpha ) & = & \neg \mathit{Tr} (\alpha),\\
\mathit{Tr} (\alpha \vee \beta) & = & \mathit{Tr} (\alpha) \vee \mathit{Tr}  (\beta).
\end{array}
$$
where $P$ is a predicate symbol and $t, t', t_1 \dots t_k$ are terms.


We will also note \textit{Tr} as its extension to sets of formulae (i.e., $\mathit{Tr}(\Gamma)$ with $\Gamma \in \mathbf{ax}(\Sigma)$).
  Now we can define $\mathsf{Struct}_{SC}$ as $\< \mathcal{O}, \mathcal{A} \>$ where 
  \begin{itemize}
  \item $\mathcal{O} = \{ \< \mathit{Str}^{\mathsf{prop}(\Sigma), \mathit{Tr}(\Gamma)}, \cup, \emptyset \> \mid \Sigma \in |\mathsf{Sign}_0|\ \land\ \Gamma \in \mathbf{ax}(\Sigma) \}$, and 
  \item $\mathcal{A} = \{ \widehat{\sigma}: \< \mathit{Str}^{\mathsf{prop}(\Sigma), \mathit{Tr}(\Gamma)}, \cup, \emptyset \> \to \< \mathit{Str}^{\mathsf{prop}(\Sigma'), \mathit{Tr}(\Gamma')}, \cup, \emptyset \> \mid \sigma : \< \Sigma, \Gamma \> \to \< \Sigma', \Gamma' \> \in || \mathsf{Th} ||\}$. 
  \end{itemize}
Finally, we can define $\mathbf{M}$ in the following way:
\begin{itemize}
 \item $\mathbf{M}(\< \Sigma, \Gamma \>) = \< \mathit{Str}^{\mathsf{prop}(\Sigma), \mathit{Tr}(\Gamma)}, \cup, \emptyset \>$, for each $\< \Sigma, \Gamma \> \in |\mathsf{Th}_{ax}|$.
 \item $\mathbf{M}(\sigma^\mathsf{op}) = \widehat{\sigma^\mathsf{op}}$, for each theory morphism $\sigma : \< \Sigma, \Gamma \> \to \< \Sigma', \Gamma' \> \in || \mathsf{Th}_{ax} ||$, where $\widehat{\sigma}: \< \mathit{Str}^{\mathsf{prop}(\Sigma), \mathit{Tr}(\Gamma)}, \cup, \emptyset \> \to \< \mathit{Str}^{\mathsf{prop}(\Sigma'), \mathit{Tr}(\Gamma')}, \cup, \emptyset \>$ is the homomorphic extension of $\sigma$.
\end{itemize}
 \item \sloppy $\mathbf{Mods}$ provides the means for obtaining the set of canonical models from structures of the form $\mathit{Str}^{\mathsf{prop}(\Sigma), \mathit{Tr}(\Gamma)}$ by identifying the "rows" of the truth table that satisfy each formula in $\Gamma$ (i.e. , the \emph{truth assignments} of the propositional variables for which the \emph{truth assignments} for each formula is $\top$). $\mathbf{Mods} : \mathsf{Struct}_\mathit{SC} \to \mathsf{Space}$ is defined on objects as $\mathbf{Mods}(\< \mathit{Str}^{\mathsf{prop}(\Sigma), \mathit{Tr}(\Gamma)}, \cup, \emptyset \>) = \< \mathcal{O}, \mathcal{A} \>$ where:
\begin{itemize}
  \item $\mathcal{O} = 
          \setof
            {\mathit{val}_{\mathsf{prop}(\Sigma)}\big|_{\mathit{Tr}(\Gamma)}}
            {
              \begin{array}{l}
                \mathit{val}_{\mathsf{prop}(\Sigma)} \in \mathit{assigns}_{\mathsf{prop}(\Sigma)}, \\
                t : \emptyset \to \mathit{Tr}(\Gamma) \in ||\mathit{Str}^{\mathsf{prop}(\Sigma), \mathit{Tr}(\Gamma)}||, \\
                (\forall \alpha \in \mathit{Tr}(\Gamma))\ t(\mathit{val}_{\mathsf{prop}(\Sigma)})(\alpha) = \top
              \end{array}
            }$
 \\ where $\mathit{val}_{\mathsf{prop}(\Sigma)}\big|_{\mathit{Tr}(\Gamma)}$ is the restriction of the function $\mathit{val}_{\mathsf{prop}(\Sigma)}$ to the set of propositional variables that appear in $\mathit{Tr}(\Gamma)$.
 \item $\mathcal{A} = \{ \mathit{id}_v : v \to v \mid v \in \mathcal{O} \}$.
\end{itemize}

\sloppy and on morphisms as for all $\sigma : \< \Sigma, \Gamma \> \to \< \Sigma', \Gamma' \> \in || \mathsf{Th}_{ax} ||, \mathbf{Mods}(\widehat{\sigma^{\mathsf{op}}})( \mathit{val}_{\mathsf{prop}(\Sigma')} \big|_{\mathit{Tr}(\Gamma')} ) = \mathit{val}_{\mathsf{prop}(\Sigma)}  \big|_{\mathit{Tr}(\Gamma)}$, such that $\mathit{val}_{\mathsf{prop}(\Sigma)} (v_p) =  \mathit{val}_{\mathsf{prop}(\Sigma')} (\sigma(v_p))$ for all the variables $v_p$  that appear in $\mathit{Tr}(\Gamma)$. Each object $\< \mathcal{O}, \mathcal{A} \>$ is a set of functions with a finite fixed domain, therefore there is at most a finite amount of functions, thus we can prove that $\< \mathcal{O}, \mathcal{A} \>$ is a space.
\item $\mu$ has to relate the structures representing truth assignments with the first order models satisfying the set of formulae denoted by those truth assignments; then we can define $\mu_{\< \Sigma, \Delta \>} : \mathbf{models}(\< \Sigma, \Delta \>) \to \mathscr{P} \circ \mathbf{Mod} (\< \Sigma, \Delta \>)$ as for each $\mathit{val}_{\mathsf{prop}(\Sigma)} \in |\mathbf{models}(\< \Sigma, \Delta \>)|, \mu_{\< \Sigma, \Delta \>}(\mathit{val}_\mathsf{prop}(\Sigma)) = \mathbf{Mod} (\< \Sigma, \widehat{\Delta} \>)$, where $\widehat{\Delta} = \{ P (t_1, \dots , t_k) \mid \mathit{val}_{\mathsf{prop}(\Sigma)}(v_{\mbox{``$P(t_1, \ldots, t_k)$''}}) = \top \}\ \cup \{ \neg P (t_1, \dots , t_k) \mid \mathit{val}_{\mathsf{prop}(\Sigma)}(v_{\mbox{``$P(t_1, \ldots, t_k)$''}}) = \bot \}\ \cup \{ t = t' \mid \mathit{val}_{\mathsf{prop}(\Sigma)}(v_{\mbox{``$t = t'$''}}) = \top \}\ \cup \{\neg t = t' \mid \mathit{val}_{\mathsf{prop}(\Sigma)}(v_{\mbox{``$t = t'$''}}) = \bot \}$.
In order to prove that $\mu$ is a natural family of functors we have to show that $\mu_{\< \Sigma, \Delta \>} \circ \mathbf{models}(\sigma^\mathsf{op})(\mathit{val}_{\mathsf{prop}(\Sigma')}) = (\mathscr{P} \circ \mathbf{Mod} (\sigma^\mathsf{op})) \circ \mu_{\< \Sigma', \Delta' \>}(\mathit{val}_{\mathsf{prop}(\Sigma')})$, for $\sigma : \< \Sigma, \Delta \> \to \< \Sigma', \Delta' \>$. It is easy to observe that:

\begin{itemize}
 \item $\mu_{\< \Sigma, \Delta \>} \circ \mathbf{models}(\sigma^\mathsf{op})(\mathit{val}_{\mathsf{prop}(\Sigma')}) =  \mathbf{Mod} (\< \Sigma, \widehat{\Delta} \>)$.
 \item $(\mathscr{P} \circ \mathbf{Mod} (\sigma^\mathsf{op})) \circ \mu_{\< \Sigma', \Delta' \>}(\mathit{val}_{\mathsf{prop}(\Sigma')}) =  \mathbf{Mod} (\< \Sigma, \widehat{\Delta} \>)$.
\end{itemize}
\noindent where $\widehat{\Delta}$, as defined above, is the set of first order predicates that correspond to the \emph{truth assignment} $\mathit{val}_{\mathsf{prop}(\Sigma)}$, and $\mathit{val}_{\mathsf{prop}(\Sigma)}$ is the map of $\mathit{val}_{\mathsf{prop}(\Sigma')}$ to the signature $\Sigma$.
\end{enumerate}

\end{example}

The notion of effective satisfiability subcalculus provides the theoretical foundations for those analysis commands relying on model finding, much in the same way that proof calculus provide theoretical support for the proof commands; the combination results from applying Coro.~\ref{coro-soundness}. Applying such a corollary requires the language to have negation, in the sense of Def.~\ref{have-negation} which, even when it might be satisfied by the institution underlying the effective satisfiability subcalculus, there is no guarantee that it is satisfied by the restriction of the language for which the subcalculus is effective.\\

\begin{proposition}\ \\
Let the structure $\< \mathsf{Sign}, \mathbf{Sen}, \mathbf{Mod},  \mathsf{Sign}_0, \mathbf{Sen}_0, \mathbf{ax}, \{\models^{\Sigma}\}_{\Sigma \in |\mathsf{Sign}|}, \mathbf{M}, \mathbf{Mods}, \mu \>$ be an effective satisfiability subcalculus, then $\< \mathsf{Sign}_0, \mathbf{Sen}_0, \mathbf{Mod}_0, \{{\models_0}^{\Sigma}\}_{\Sigma \in |\mathsf{Sign}_0|}\>$ is an institution, where $\mathbf{Mod}_0$ stands for the functor $\mathbf{Mod}$ restricted to the subcategory ${\mathsf{Sign}_0}^\op$ (usually denoted as ${\mathbf{Mod}}\left|\right._{{\mathsf{Sign}_0}^\op}$) and ${\models_0}^\Sigma$ for the restriction of $\models^\Sigma$ to the models in $\mathbf{Mod}_0 (\Sigma)$ and the formulae in $\mathbf{Sen}_0 (\Sigma)$ (usually denoted as ${\models_0}^\Sigma = {\models^\Sigma}\left|\right._{\mathbf{Mod}_0 \times \mathbf{Sen}_0}$).
\end{proposition}

We already shown how Coro.~\ref{coro-soundness} can be used to conclude the unprovability of a given formula $\alpha$ from a theory presentation $T$, by using a satisfiability calculus in order to construct a structure $M \in |\mathbf{models}(T)|$ and then, having to find a model $\mathcal{M}$ such that $\mathcal{M} \models T$ and $\mathcal{M} \models \alpha$ does not hold which, under the presence of negation in the syntax of the logical language (see Def.~\ref{have-negation}), is equivalent to construct a model $\mathcal{M}$ such that $\mathcal{M} \models T$ and $\mathcal{M} \models \neg\alpha$. Effectiveness, as it was introduced by Def.~\ref{eff-sat-subcalculus}, states that $M$ is the result of a constructive push-button process but, for that to hold, the language of the effective satisfiability subcalculus must have negation (i.e. the institution $\< \mathsf{Sign}_0, \mathbf{Sen}_0, \mathbf{Mod}_0, \{{\models_0}^{\Sigma}\}_{\Sigma \in |\mathsf{Sign}_0|}\>$ has to have negation). 

\begin{corollary}
\label{coro-sub-soundness}
Let the structure $\< \mathsf{Sign}, \mathbf{Sen}, \mathbf{Mod}, \{\vdash^{\Sigma}\}_{\Sigma \in |\mathsf{Sign}|}, \{\models^{\Sigma}\}_{\Sigma \in |\mathsf{Sign}|} \>$ be a logic, the structure $\< \mathsf{Sign}, \mathbf{Sen}, \{\vdash^{\Sigma}\}_{\Sigma \in |\mathsf{Sign}|}, \mathbf{P}, \mathbf{Pr}, \pi \>$ be a proof calculus and the structure $\< \mathsf{Sign}, \mathbf{Sen}, \mathbf{Mod},  \mathsf{Sign}_0, \mathbf{Sen}_0, \mathbf{ax}, \{\models^{\Sigma}\}_{\Sigma \in |\mathsf{Sign}|}, \mathbf{M}, \mathbf{Mods}, \mu \>$ be an effective satisfiability subcalculus such that $\< \mathsf{Sign}_0, \mathbf{Sen}_0, \mathbf{Mod}_0, \{{\models_0}^{\Sigma}\}_{\Sigma \in |\mathsf{Sign}|}\>$ has negation. Therefore, for all $T = \< \Sigma, \Gamma \> \in |\mathsf{Th}_0|$ and $\alpha \in |\mathbf{Sen}_0 (\Sigma)|$, if there exists $M \in |\mathbf{models}(T)|$ and $\mathcal{M} \in |\mu_T (M)|$ such that $\mathcal{M} \models^\Sigma \neg\alpha$, then there is no $\tau \in |\mathbf{proof}(T)|$ such that $\pi_T(\tau) = \alpha$.
\end{corollary}

As we mentioned at the beginning of this section, our focus is in combining automatic model finding capabilities as part of the process of analysis of software descriptions. While Def.~\ref{eff-sat-subcalculus} provides the formal framework for describing such effective procedures and, as we said before, Coro.~\ref{coro-sub-soundness} expose the link between the result of such procedures and the unprovability of a formula, there is still a missing piece in the puzzle in order to guarantee automatic counterexample finding. 

Corollary~\ref{coro-sub-soundness} properly connects the existence of a model satisfying the negation of a formula with the impossibility of proving the formula itself, but it provides no guide on how to use the effective procedure for building a model structure $M \in |\mathbf{models}(T)|$ for finding a model $\mathcal{M} \in |\mu_T (M)|$ such that $\mathcal{M} \models^\Sigma \neg\alpha$. To solve that, $\neg\alpha$ has to be involved in the process of building $M$ in a way that every model $\mathcal{M} \in |\mu_T (M)|$ satisfies $\neg\alpha$. In general, the mechanism used to internalise $\neg\alpha$ in the process for building models for $T$ strongly depend on the specific procedure for building model structures and, under certain conditions, it can be done by simply adding $\neg\alpha$ to the axioms in $T$. Such a condition is know, analogous to what we did for requiring an institution to have negation, as to ``have conjunction'', formally defined below.

\begin{definition}
\label{have-conjunction}
Given a institution $\mathbb{I}$, we say that \emph{$\mathbb{I}$ has conjunction} if for every signature $\Sigma \in |\mathsf{Sign}|$, $\{\alpha_i\}_{i \in \mathcal{I}} \subseteq \mathbf{Sen} (\Sigma)$, there exists a formula in $\mathbf{Sen} (\Sigma)$, usually denoted as $\bigwedge_{i \in \mathcal{I}} \alpha_i$, such that for all $\mathcal{M} \in \mathbf{Mod} (\Sigma)$,
$$\mathcal{M} \models^\Sigma \bigwedge_{i \in \mathcal{I}} \alpha_i \text{ iff } \mathcal{M} \models^\Sigma \alpha_i \text{, for all } i \in \mathcal{I}$$
\end{definition}

\begin{remark}
Notice that given a institution $\mathbb{I}$, $\Sigma \in |\mathsf{Sign}|$ and $\Gamma \subseteq \mathbf{Sen} (\Sigma)$, for all $\mathcal{M} \in \mathbf{Mod} (\Sigma)$,
$$\mathcal{M} \models^\Sigma \Gamma \text{ iff } \mathcal{M} \models^\Sigma \bigwedge_{\alpha \in \Gamma} \alpha$$
\end{remark}

Therefore, Coro.~\ref{coro-sub-soundness} can be stated as follows.

\begin{corollary}
\label{coro-sub-soundness-bmc}
Let the structure $\< \mathsf{Sign}, \mathbf{Sen}, \mathbf{Mod}, \{\vdash^{\Sigma}\}_{\Sigma \in |\mathsf{Sign}|}, \{\models^{\Sigma}\}_{\Sigma \in |\mathsf{Sign}|} \>$ be a logic, the structure $\< \mathsf{Sign}, \mathbf{Sen}, \{\vdash^{\Sigma}\}_{\Sigma \in |\mathsf{Sign}|}, \mathbf{P}, \mathbf{Pr}, \pi \>$ be a proof calculus and the structure $\< \mathsf{Sign}, \mathbf{Sen}, \mathbf{Mod},  \mathsf{Sign}_0, \mathbf{Sen}_0, \mathbf{ax}, \{\models^{\Sigma}\}_{\Sigma \in |\mathsf{Sign}|}, \mathbf{M}, \mathbf{Mods}, \mu \>$ be an effective satisfiability subcalculus such that $\< \mathsf{Sign}_0, \mathbf{Sen}_0, \mathbf{Mod}_0, \{{\models_0}^{\Sigma}\}_{\Sigma \in |\mathsf{Sign}_0|}\>$ has negation and conjunction. Therefore, for all $T = \< \Sigma, \Gamma \> \in |\mathsf{Th}_0|$ and $\alpha \in |\mathbf{Sen}_0 (\Sigma)|$, if there exists $M \in |\mathbf{models}(\< \Sigma, \Gamma \cup \{\neg\alpha\} \>)|$, then there is no $\tau \in |\mathbf{proof}(T)|$ such that $\pi_T(\tau) = \alpha$.
\end{corollary}

So far, we managed to provide the means for understanding how an effective procedure for building models can be used for short-cutting a proving attempt by deducing that no such proof exists but such usage is confined to theories and propositions consisting exclusively of formulae from the sublanguage formalised as the institution $\< \mathsf{Sign}_0, \mathbf{Sen}_0, \mathbf{Mod}_0, \{{\models_0}^{\Sigma}\}_{\Sigma \in |\mathsf{Sign}_0|}\>$ for which such an effective procedure exists. We would like to get as much value as possible from the, usually big, effort invested in developing tools supporting specific effective satisfiability calculi. 

\sloppy
The typical way of doing so is as follows: given an institution of interest $\mathbbm{I} = \< \mathsf{Sign}, \mathbf{Sen}, \mathbf{Mod}, \{\models^{\Sigma}\}_{\Sigma \in |\mathsf{Sign}|}\>$ and an effective satisfiability subcalculus $\mathbbm{Q} = \< \mathsf{Sign}', \mathbf{Sen}', \mathbf{Mod}',  \mathsf{Sign}'_0, \mathbf{Sen}'_0, \mathbf{ax}', \{{\models'}^{\Sigma}\}_{\Sigma \in |\mathsf{Sign}'|}, \mathbf{M}', \mathbf{Mods}', \mu' \>$, we consider a semantics preserving translation from $\mathbbm{I}$ to $\< \mathsf{Sign}'_ 0, \mathbf{Sen}'_0, \mathbf{Mod}'_0, \{{\models'_ 0}^{\Sigma}\}_{\Sigma \in |\mathsf{Sign}'_0|} \>$, the institution underlying the decidable fragment of $\mathbbm{Q}$.

The usual notion of morphism between institutions consists of three arrows between the components of the source and target institutions:
\begin{inparaenum}[1)]
\item an arrow relating the categories of signatures,
\item an arrow relating the grammar functors, and
\item an arrow relating the models functors.
\end{inparaenum} 
Several notions of morphisms between institutions, and their properties, were investigated in \cite{goguen:jacm-39_1,meseguer:lc87,tarlecki:sadt-rtdts95}. More recently, in \cite{goguen:fac-13_3-5}, all these notions of morphism were scrutinised in more detail by observing how the direction of the arrows modify its interpretation. Institution comorphisms can be used as a vehicle for borrowing proofs along logic translation \cite{tarlecki:sadt-rtdts95}; for defining heterogeneous development environments for software specification and design \cite{mossakowski:wadt2008,lopezpombo:ictac14}, which provides the foundations of tools like HETS \cite{mossakowski:tacas07} and CafeOBJ \cite{diaconescu:tcs-285_2}; for providing structured specifications in general \cite{sannella:ic-76_2-3}, for providing semantics to specific features of formal languages \cite{castro:fac-27_5-6}; for defining proof systems for structured specifications \cite{borzyszkowski:tcs-286_2,mossakowski:fossacs14,lopezpombo:ocl17}; and for formalising data and specification refinements in \cite{tarlecki:ctcp86,castro:ictac10,castro:facs12}, just to give a few examples.

Let us start by recalling some notions about mapping entailment systems. In general, given two entailment systems $\mathbb{E} = \< \mathsf{Sign}, \mathbf{Sen}, \{\vdash^{\Sigma}\}_{\Sigma \in |\mathsf{Sign}|} \>$ and $\mathbb{E}' = \< \mathsf{Sign}', \mathbf{Sen}', \{{\vdash'}^{\Sigma}\}_{\Sigma \in |\mathsf{Sign}|} \>$, if we aim at mapping sentences in $\mathbb{E}$ to sentences in $\mathbb{E}'$ by using a natural transformation $\alpha: \mathbf{Sen} \nattrans \mathbf{Sen}' \circ \Phi$ in a way that provability is preserved (i.e. $\Gamma \vdash^\Sigma \phi$ if and only if $\alpha_\Sigma (\Gamma) {\vdash'}^{\Phi (\Sigma)} \alpha_\Sigma (\phi)$), a functor $\Phi: \mathsf{Sign} \to \mathsf{Sign}'$ will not be enough. The reason is that for preserving provability, the proof theoretic structure of $\mathbbm{E}'$ might require certain basic axioms for characterising the behaviour of its own logical symbols. Solving this issue requires resorting to a functor $\Phi: \mathsf{Sign} \to \mathsf{Th}'_0$ that, in turn, can be easily extended to a functor $\overline{\Phi}: \mathsf{Th}_0 \to \mathsf{Th}'_0$, the \emph{$\alpha$-extension to theories} of $\Phi$, as $\overline{\Phi} (\<\Sigma, \Gamma\>) = \<\Sigma', \overline{\Gamma'}\>$ such that $\Phi (\Sigma) = \<\Sigma', \Gamma'\>$ and $\overline{\Gamma'} = \Gamma' \cup \alpha_\Sigma (\Gamma)$. Notice also that the natural transformation $\alpha: \mathbf{Sen} \nattrans \mathbf{Sen}' \circ \Phi$ can be extended to $\overline{\alpha}: \mathbf{Sen} \nattrans \mathbf{Sen}' \circ \Phi$ between the functors $\mathbf{Sen}: \mathsf{Th}_0 \to \mathsf{Set}$ and $\mathbf{Sen}' \circ \overline{\Phi}: \mathsf{Th}'_0 \to \mathsf{Set}$ as $\alpha_{\<\Sigma, \Gamma\>} = \alpha_{\Sigma}$.

\begin{definition}[Sensibility and simplicity \cite{meseguer:lc87}]
\label{sensibility}
Let $\< \mathsf{Sign}, \mathbf{Sen}, \{\vdash^{\Sigma}\}_{\Sigma \in |\mathsf{Sign}|} \>$ and $\<\mathsf{Sign}', \mathbf{Sen}', \{{\vdash'}^{\Sigma}\}_{\Sigma \in |\mathsf{Sign}'|} \>$ be entailment systems, $\Phi: {\sf Th}_0 \to {\sf Th'}_0$ be a functor and $\alpha: \mathbf{Sen} \nattrans \mathbf{Sen}' \circ \Phi$ a natural transformation. $\Phi$ is said to be \emph{$\alpha$-sensible} if the following conditions are satisfied:
\begin{enumerate}
\item \label{alpha-cond-1} there is a functor $\Phi^\diamond: \mathsf{Sign} \to \mathsf{Sign}'$ such that ${\bf sign'} \circ \Phi = \Phi^\diamond \circ {\bf sign}$, where ${\bf sign}: \mathsf{Th} \to \mathsf{Sign}$ and ${\bf sign}': \mathsf{Th}' \to \mathsf{Sign}'$ are the forgetful functors from theory presentations to signatures, and
\item \label{alpha-cond-2} if $\<\Sigma, \Gamma\> \in |{\sf Th}|$ and $\<\Sigma', \Gamma'\> \in |{\sf Th'}|$ such that $\Phi (\<\Sigma, \Gamma\>) = \<\Sigma', \Gamma'\>$, then $(\Gamma')^\bullet = (\emptyset' \cup \alpha_\Sigma (\Gamma))^\bullet$, where $\emptyset' = \alpha_\Sigma (\emptyset)$\footnote{$\emptyset'$ is not necessarily the empty set of axioms. This fact will be clarified later on.}.
\end{enumerate}
$\Phi$ is said to be \emph{$\alpha$-simple} if, instead of satisfying $(\Gamma')^\bullet = (\emptyset' \cup \alpha_\Sigma (\Gamma))^\bullet$ in Condition~\ref{alpha-cond-2}, the stronger condition $\Gamma' = \emptyset' \cup \alpha_\Sigma (\Gamma)$ is satisfied. 
\end{definition}

It is straightforward to see, based on the monotonicity of $^\bullet$, that $\alpha$-simplicity implies $\alpha$-sensibility. An $\alpha$-sensible functor has the property that the associated natural transformation $\alpha$ depends only on signatures. This is a consequence of the following lemma.

\begin{lemma}[Lemma~22, \cite{meseguer:lc87}]\ \\
Let $\< \mathsf{Sign}, \mathbf{Sen}, \{\vdash^{\Sigma}\}_{\Sigma \in |\mathsf{Sign}|} \>$ and $\<\mathsf{Sign}', \mathbf{Sen}', \{{\vdash'}^{\Sigma}\}_{\Sigma \in |\mathsf{Sign}'|} \>$ be entailment systems and $\Phi: {\sf Th}_0 \to {\sf Th'}_0$ a functor satisfying Cond.~\ref{alpha-cond-1}~of~Def.~\ref{sensibility}. Then any natural transformation $\alpha: \mathbf{Sen} \nattrans \mathbf{Sen}' \circ \Phi$ can be obtained from a natural transformation $\alpha^\diamond: \mathbf{Sen}(\Sigma) \nattrans \mathbf{Sen}' \circ \Phi^\diamond$ by horizontally composing with the functor ${\bf sign}: {\sf Th} \to {\sf Sign}$.
\end{lemma}

The next definition introduces institution comorphism.

\begin{definition}[Institution comorphism \cite{tarlecki:sadt-rtdts95}]\ \\
\label{inst-co-morphism}
Let $\mathbb{I}$ and $\mathbb{I}'$ be institutions. Then, $\< \rho^{Sign}, \rho^{Sen}, \rho^{Mod} \>: \mathbb{I} \longrightarrow \mathbb{I}'$ is an \emph{institution comorphism} if and only if:
\begin{itemize}
\item $\rho^{Sign}: \mathsf{Sign} \to \mathsf{Sign}'$ is a functor,
\item $\rho^{Sen}: \mathbf{Sen} \nattrans \mathbf{Sen}' \circ \rho^{Sign}$ is a natural transformation, (i.e., a natural family of functions $\rho^{Sen}_{\Sigma}: \mathbf{Sen}(\Sigma) \to \mathbf{Sen}'(\rho^{Sign}(\Sigma))$), such that for each $\Sigma_1, \Sigma_2 \in |\mathsf{Sign}|$ and $\sigma: \Sigma_1 \to \Sigma_2$ morphism in $\mathsf{Sign}$, the following diagram commutes
\unitlength1cm
\begin{center}
\begin{picture}(11,3)(0,-3)
\put(0,-3){}
\put(0.3, -0.8){$\mathbf{Sen}(\Sigma_2)$}
\put(0.8,-2.4){\vector(0,1){1.4}}
\put(-0.4,-1.8){$\mathbf{Sen}(\sigma)$}
\put(0.3,-2.8){$\mathbf{Sen}(\Sigma_1)$}
\put(1.7,-0.7){\vector(1,0){2.4}}
\put(2.6,-0.4){$\rho^{Sen}_{\Sigma_2}$}
\put(4.3,-0.8){$\mathbf{Sen}'(\rho^{Sign} (\Sigma_2))$}
\put(5.2,-2.4){\vector(0,1){1.4}}
\put(5.3,-1.8){$\mathbf{Sen}'(\rho^{Sign} (\sigma))$}
\put(4.3,-2.8){$\mathbf{Sen}'(\rho^{Sign} (\Sigma_1))$}
\put(1.7,-2.7){\vector(1,0){2.4}}
\put(2.6,-2.4){$\rho^{Sen}_{\Sigma_1}$}

\put(10, -0.8){$\Sigma_2$}
\put(10.2,-2.4){\vector(0,1){1.4}}
\put(10.4,-1.8){$\sigma$}
\put(10,-2.8){$\Sigma_1$}
\end{picture}
\end{center}
\item $\rho^{Mod}: \mathbf{Mod}' \circ {\rho^{Sign}}^\op \nattrans \mathbf{Mod}$\footnote{The functor ${\rho^{Sign}}^\op: {\mathsf{Sign}'}^\op \to \mathsf{Sign}^\op$ is the same as $\rho^{Sign}: \mathsf{Sign}' \to \mathsf{Sign}$ but considered between the opposite categories.} is a natural transformation, (i.e., the family of functors $\rho^{Mod}_{\Sigma}: \mathbf{Mod}'({\rho^{Sign}}^\op (\Sigma)) \to \mathbf{Mod}(\Sigma)$ is natural), such that for each $\Sigma_1, \Sigma_2 \in |\mathsf{Sign}|$ and $\sigma: \Sigma_1 \to \Sigma_2$ a morphism in $\mathsf{Sign}$, the following diagram commutes
\unitlength1cm
\begin{center}
\begin{picture}(11,3)(-1,-3)
\put(0,-3){}

\put(0.7, -0.7){$\mathbf{Mod}'({\rho^{Sign}}^\op (\Sigma_2))$}
\put(1.9,-0.95){\vector(0,-1){1.4}}
\put(-1.3,-1.8){$\mathbf{Mod}'({\rho^{Sign}}^\op (\sigma^\op))$}
\put(0.75,-2.8){$\mathbf{Mod}'({\rho^{Sign}}^\op (\Sigma_1))$}
\put(3.8,-0.65){\vector(1,0){2.4}}
\put(4.65,-0.4){$\rho^{Mod}_{\Sigma_2}$}
\put(6.35,-0.7){$\mathbf{Mod}(\Sigma_2)$}
\put(7.05,-0.95){\vector(0,-1){1.4}}
\put(7.15,-1.8){$\mathbf{Mod}(\sigma^\op)$}
\put(6.35,-2.8){$\mathbf{Mod}(\Sigma_1)$}
\put(3.8,-2.7){\vector(1,0){2.4}}
\put(4.65,-2.4){$\rho^{Mod}_{\Sigma_1}$}

\put(9.6, -0.8){$\Sigma_2$}
\put(9.8,-2.4){\vector(0,1){1.4}}
\put(10,-1.8){$\sigma$}
\put(9.6,-2.8){$\Sigma_1$}
\end{picture}
\end{center}
\end{itemize}
such that for any $\Sigma \in |\mathsf{Sign}|$,
the function $\rho^{Sen}_{\Sigma}: \mathbf{Sen}(\Sigma) \to
  \mathbf{Sen}'(\rho^{Sign}(\Sigma))$
and the functor
$\rho^{Mod}_{\Sigma}:
  \mathbf{Mod}'({\rho^{Sign}}^\op(\Sigma)) \to \mathbf{Mod}(\Sigma)$
preserves the following satisfaction condition: for any $\alpha \in \mathbf{Sen}(\Sigma)$ and
$\mathcal{M}' \in |\mathbf{Mod}'({\rho^{Sign}}^\op(\Sigma))|$,
$$
\mathcal{M}' {\models'}^{\rho^{Sign}(\Sigma)} \rho^{Sen}_{\Sigma}
  (\alpha)\ \mbox{ iff }\ \rho^{Mod}_{\Sigma}(\mathcal{M}')
  \models^{\Sigma} \alpha\ .
$$
\end{definition}

Intuitively, an institution comorphism $\rho: \mathbb{I} \longrightarrow \mathbb{I}'$ expresses how the satisfiability in a (potentially) less expressive institution $\mathbb{I}$ ca be encoded into a (potentially) more expressive institution $\mathbb{I}'$. In \cite[Sec.~5.2]{tarlecki:sadt-rtdts95}, Tarlecki shows that institution comorphisms compose in a rather obvious component wise way, and that it can be proved that institutions together with institution comorphisms form a category, denoted as $\mathsf{coIns}$.


The following results, presented in \cite{tarlecki:sadt-rtdts95}, characterise the relation between the satisfaction relation of $\mathbb{I}$ and $\mathbb{I}'$, in the presence of an institution comorphism.

\begin{proposition}[Preservation of consequence \cite{tarlecki:sadt-rtdts95}] 
\label{preservation}
Let  $\mathbb{I}$ and $\mathbb{I}'$ be the institutions $\< \mathsf{Sign}, \mathbf{Sen}, \mathbf{Mod}, \{\models^{\Sigma}\}_{\Sigma \in |\mathsf{Sign}|}  \>$ and $\< \mathsf{Sign}', \mathbf{Sen}', \mathbf{Mod}', \{{\models'}^{\Sigma}\}_{\Sigma \in |\mathsf{Sign}'|}  \>$, respectively. Let $\rho: \mathbb{I} \to \mathbb{I}'$ be an institution comorphism. Then, for all $\Sigma \in |\mathsf{Sign}|$, $\Gamma \subseteq \mathbf{Sen}(\Sigma)$ and $\varphi \in \mathbf{Sen} (\Sigma)$, if $\Gamma \models^\Sigma \varphi$, then $\rho^{Sen}_\Sigma (\Gamma) {\models'}^{\rho^{Sign} (\Sigma)} \rho^{Sen}_\Sigma (\varphi)$.
\end{proposition}

In some cases institution comorphisms can be extended to what Meseguer introduced under the name \emph{map of institutions} \cite[Def.~27]{meseguer:lc87}, and were more recently renamed by Goguen and Ro\c{s}u as \emph{theoroidal comorphism} \cite[Def.~5.3]{goguen:fac-13_3-5}, by reformulating the functor mapping signatures to act on theory presentations (i.e. $\rho^{Th}: \mathsf{Th} \to \mathsf{Th}'$). 

This extension can be done in many ways being one of them when $\rho^{Th}$ is $\rho^{Sen}$-sensible (see Def.~\ref{sensibility}) with respect to the \emph{entailment system} induced by the consequence relations of the institutions $\mathbb{I}$ and $\mathbb{I}'$. In this way, the natural transformation $\rho^{Sign}$ relating the categories of signatures of both logical languages, can be extended to $\rho^{Th}: \mathsf{Sign} \to \mathsf{Th}'$ enabling the restriction of the target class of models over which the consequence is preserved, and then, to a functor $\rho^{Th}: \mathsf{Th} \to \mathsf{Th}'$ defined in the following way: let $\Sigma \in |\mathsf{Sign}|$, $\Gamma \in |\mathbf{Sen}(\Sigma)|$, $\rho^{Th} (\<\Sigma, \Gamma\>) = \<\rho^{Sign} (\Sigma), \emptyset' \cup \rho^{Sen}_\Sigma (\Gamma)\>$. Additionally, given an institution comorphism $\< \rho^{Sign}, \rho^{Sen}, \rho^{Mod} \>: \mathbb{I} \longrightarrow \mathbb{I}'$ we say it is \emph{plain} if and only if $\rho^{Th}$ is $\rho^{Sen}$-plain, and similarly we say it is \emph{simple} if and only if $\rho^{Th}$ is $\rho^{Sen}$-simple.

The interested reader is pointed to \cite[Sec.~4]{meseguer:lc87}, \cite[Sec.~5]{tarlecki:sadt-rtdts95} and \cite[Sec.~5]{goguen:fac-13_3-5} for a thorough discussion on the matter.\\

Therefore, the combination of Prop.~\ref{preservation} and Coro.~\ref{coro-sub-soundness-bmc} results in the final statement of how soundness enables the use of an effective satisfiability subcalculus for a logical language, as a counterexample finder for a (potentially) different one.

\begin{corollary}
\label{coro-sub-soundness-bmc-het}
Let the structure $\< \mathsf{Sign}, \mathbf{Sen}, \mathbf{Mod}, \{\vdash^{\Sigma}\}_{\Sigma \in |\mathsf{Sign}|}, \{\models^{\Sigma}\}_{\Sigma \in |\mathsf{Sign}|} \>$ be a logic, the structure $\< \mathsf{Sign}, \mathbf{Sen}, \{\vdash^{\Sigma}\}_{\Sigma \in |\mathsf{Sign}|}, \mathbf{P}, \mathbf{Pr}, \pi \>$ be a proof calculus, the structure $\< \mathsf{Sign}', \mathbf{Sen}', \mathbf{Mod}',  \mathsf{Sign}'_0, \mathbf{Sen}'_0, \mathbf{ax}', \{{\models'}^{\Sigma}\}_{\Sigma \in |\mathsf{Sign}'|}, \mathbf{M}', \mathbf{Mods}', \mu' \>$ be an effective satisfiability subcalculus such that $\< \mathsf{Sign}'_0, \mathbf{Sen}'_0, \mathbf{Mod}'_0, \{{{\models'}_0}^{\Sigma}\}_{\Sigma \in |\mathsf{Sign}'_0|}\>$ has negation and conjunction and $\rho = \<\rho^{Sign}, \rho^{Sen}, \rho^{Mod}\>: \< \mathsf{Sign}, \mathbf{Sen}, \mathbf{Mod}, \{\models^{\Sigma}\}_{\Sigma \in |\mathsf{Sign}|} \> \to \< \mathsf{Sign}'_0, \mathbf{Sen}'_0, \mathbf{Mod}'_0, \{{{\models'}_0}^{\Sigma}\}_{\Sigma \in |\mathsf{Sign}'_0|}\>$ a comorphism between institutions, being $\rho^{Th}: \mathsf{Th}_0 \to \mathsf{Th}'_0$ the theoroidal extension of $\rho^{Sign}: \mathsf{Sign} \to \mathsf{Sign}'_0$. Therefore, for all $T = \< \Sigma, \Gamma \> \in |\mathsf{Th}_0|$ and $\alpha \in |\mathbf{Sen}_0 (\Sigma)|$, if there exists $M \in |\mathbf{models}'(\rho^{Th} (\< \Sigma, \Gamma \cup \{\neg\alpha\} \>))|$, then there is no $\tau \in |\mathbf{proof}(T)|$ such that $\pi_T(\tau) = \alpha$.
\end{corollary}

Corollary~\ref{coro-sub-soundness-bmc-het} provides the formal support for using a model finder to look for counterexamples of judgements of the shape $\Gamma \vdash^\Sigma \alpha$ at any point of a proof, enabling the capability to check whether the application of a given proof command has modified the provability of $\alpha$ from $\Gamma$.

The reader should note that Coro.~\ref{coro-sub-soundness-bmc-het} only guaranties that, given a theory presentation $T = \<\Sigma, \Gamma\>$ and a formula $\alpha$, whenever a model structure $M$ is built for $\rho^{Th} (\<\Sigma, \Gamma \cup \{\neg\alpha\}\>)$ using an effective satisfiability subcalculus, we can conclude the impossibility of constructing a proof structure for $\alpha$ from $\Gamma$; but not being able to do so does not provide any solid information about its existence, because there could still be a model in the source institution satisfying the axioms in $T$ and not satisfying $\alpha$. That, of course, does not invalidate the confidence gained in the validity of the judgement $\gamma \models^\Sigma \alpha$, by not founding a counterexample.\\

Depending on how precise the natural transformation $\rho^{Mod}$ is, regarding the representation of the target category of models over the source one, we can derive stronger conclusion.

\begin{definition}[Model expansion over model translation]
\label{rho-expansion}
Let  $\mathbb{I}$ and $\mathbb{I}'$ be the institutions $\< \mathsf{Sign}, \mathbf{Sen}, \mathbf{Mod}, \{\models_{\Sigma}\}_{\Sigma \in |\mathsf{Sign}|}  \>$ and $\< \mathsf{Sign}', \mathbf{Sen}', \mathbf{Mod}', \{{\models'}^{\Sigma}\}_{\Sigma \in |\mathsf{Sign}'|}  \>$, respectively. Let $\rho: \mathbb{I} \to \mathbb{I}'$ be an institution comorphism. Then, $\mathbb{I}$ has the \emph{$\rho$-expansion property} if for all $\<\Sigma, \Gamma\> \in |\mathsf{Th}^\mathbb{I}_0|$, $\mathcal{M} \in |\mathbf{Mod}^\mathbb{I}(\<\Sigma, \Gamma\>)|$, there exists $\mathcal{M}' \in |\mathbf{Mod}^{\mathbb{I}'} (\rho^{Th} (\< \Sigma, \Gamma\>))|$ such that $\mathcal{M} = \rho^{Mod} (\mathcal{M}')$.
\end{definition}

The intuition behind the previous definition is that, given a theoroidal institution comorphism $\rho: \mathbbm{I} \to \mathbbm{I}'$, $\mathbbm{I}$ has the $\rho$-expansion property whenever every $\mathbbm{I}$-model is the target of some $\mathbbm{I}'$-models, in a way that the satisfaction of $\mathbbm{I}$-formulae is reflected in $\mathbbm{I}'$.

\begin{theorem}[Reflection of consequence \cite{tarlecki:sadt-rtdts95}]
\label{tarlecki-thm}
Let  $\mathbb{I}$ and $\mathbb{I}'$ be the institutions $\< \mathsf{Sign}, \mathbf{Sen}, \mathbf{Mod}, \{\models_{\Sigma}\}_{\Sigma \in |\mathsf{Sign}|}  \>$ and $\< \mathsf{Sign}', \mathbf{Sen}', \mathbf{Mod}', \{{\models'}^{\Sigma}\}_{\Sigma \in |\mathsf{Sign}'|}  \>$, respectively. Let $\rho: \mathbb{I} \to \mathbb{I}'$ be an institution comorphism. Then, for all $\Sigma \in |\mathsf{Sign}|$, $\Gamma \subseteq \mathbf{Sen}(\Sigma)$ and $\varphi \in \mathbf{Sen} (\Sigma)$, if every $\mathcal{M} \in \mathbf{Mod}(\<\Sigma, \Gamma\>)$ has the $\rho$-expansion property, then $\Gamma \models_\Sigma \varphi$ if and only if $\rho^{Sen}_\Sigma (\Gamma) \models_{\rho^{Sign} (\Sigma)} \rho^{Sen}_\Sigma (\varphi)$.
\end{theorem}

Finally, Thm.~\ref{tarlecki-thm} states that, given a theoroidal institution comorphism $\rho: \mathbbm{I} \to \mathbbm{I}'$, such that $\mathbbm{I}$ has the $\rho$-expansion property, then an effective satisfiability subcalculus for $\mathbbm{I}'$ (i.e. $\mathbbm{I}'$ is the fragment of the language of the effective satisfiability subcalculus for which the procedure is indeed effective) constitutes a semantics-based decision procedure for $\mathbbm{I}$.


\section{Case-study: first-order predicate logic}
\label{case-study}

First-order predicate logic provides a well-known, and simple, example for showing the value of using theoroidal comorphisms as a tool for reusing tools designed for one language in the analysis of software specifications written in ``another''. Let us consider the effective satisfiability subcalculus, presented in Ex.~\ref{ex:eff-sat-subcalculus}, for the quantifier-free and ground fragment of $\mathit{FOL}$. Intuitively, the theories for which we are capable of constructing models for, are those that do not resort to any quantification and, of course, have no free variables (i.e. all terms appearing in the formulae are ground), but Coro.~\ref{coro-sub-soundness-bmc-het} provides the mathematical machinery needed for using off-the-shelve constraint-solving tools in order to enable counterexample searching capabilities.\\

The reader might note that such a semantics preserving translation contradicts the intuition mentioned after Def.~\ref{inst-co-morphism} where we slide the misleading, but widely spread, idea that the source logical system of an institution comorphism is less expressive than the target one. A more precise interpretation is to consider institution comorphisms as encodings that, in general, satisfy the intuition mentioned before but, under some specific conditions, might loose expressivity / representation capability, while, as we will show in the forthcoming section, still preserve the satisfaction condition. A more thorough study of the conditions under which this type of institution comorphisms can be defined will be left for further investigation.\\

\sloppy Let $\mathcal{X}$ be a set of first order variable, $\mathbbm{L}_\text{FOL} = \< \mathsf{Sign}, \mathbf{Sen}, \mathbf{Mod}, \{\vdash^{\Sigma}\}_{\Sigma \in |\mathsf{Sign}|}, \{\models^{\Sigma}\}_{\Sigma \in |\mathsf{Sign}|} \>$ be the structure formalising first-order predicate logic over $\mathcal{X}$ and $\mathbbm{P}_\text{FOL} = \< \mathsf{Sign}, \mathbf{Sen}, \{\vdash^{\Sigma}\}_{\Sigma \in |\mathsf{Sign}|}, \mathbf{P}, \mathbf{Pr}, \pi \>$ be a proof calculus for $\mathbbm{L}$.

\sloppy Let $\mathbbm{Q}_\text{Prop}$ be the effective satisfiability subcalculus $\< \mathsf{Sign}', \mathbf{Sen}', \mathbf{Mod}',  \mathsf{Sign}'_0, \mathbf{Sen}'_0, \mathbf{ax}', \{{\models'}^{\Sigma}\}_{\Sigma \in |\mathsf{Sign}'|}, \mathbf{M}', \mathbf{Mods}', \mu' \>$, where $\< \mathsf{Sign}', \mathbf{Sen}', \mathbf{Mod}', \{{\models'}^{\Sigma}\}_{\Sigma \in |\mathsf{Sign}'|} \>$ is the institution of propositional logic\footnote{Note that in this particular case $\< \mathsf{Sign}', \mathbf{Sen}', \mathbf{Mod}', \{{\models'}^{\Sigma}\}_{\Sigma \in |\mathsf{Sign}'|} \>$ and $\< \mathsf{Sign}'_0, \mathbf{Sen}'_0, \mathbf{Mod}'_0, \{{\models'}_0^{\Sigma}\}_{\Sigma \in |\mathsf{Sign}'_0|} \>$ can be assumed to be the same institution due to the fact that propositional logic is decidable.}. Then, we can define a theoroidal comorphism from the underlying institution of $\mathbbm{L}_\text{FOL}$ to the underlying institution of $\mathbbm{Q}_\text{Prop}$ by parameterising the translation of first-order formulae with a natural number used to enforce finite bounds to the domains of interpretation of first-order formulae.

\begin{definition} 
\label{gamma-sign}
$\gamma^{Sign}: \mathsf{Sign} \to \mathsf{Sign}'$ is defined as the functor such that:
\begin{itemize}  
\item for all $\Sigma \in |\mathsf{Sign}|$, $\gamma^{Sign}(\Sigma) = \{v_\mathit{p} | \mathit{p}$ is a ground atomic formula in $\mathbf{Sen} (\Sigma) \}$, and
\item for all $\sigma: \Sigma \to \Sigma' \in ||\mathsf{Sign}||$, $\gamma^{Sign}(\sigma) = \tau$, where $\tau(v_p) = v_{\mathbf{Sen} (\sigma)(p)}$.
\end{itemize}
\end{definition}

\begin{proposition}
    \label{gamma-sign-proof}
    $\gamma^{Sign}$ is a functor.
\end{proposition}
\begin{proof}
Note that $\gamma^{Sign}(\sigma)$ is defined in terms of $\mathbf{Sen} (\sigma)$. Since $\mathbf{Sen}$ is a functor it can be proved that $\gamma^{Sign}$ preserves identities and composition.
\end{proof}

\begin{definition}
\label{gamma-sen}
Let $n \in \Nat$, $\Sigma = \<\{f_i\}_{i \in \mathcal{I}}, \{P_j\}_{j \in \mathcal{J}}\>$ and $\Sigma' = \<\{f_i\}_{i \in \mathcal{I}} \cup \{c_k\}_{1 \leq k \leq n}, \{P_j\}_{j \in \mathcal{J}}\>$ such that $\Sigma, \Sigma' \in |\mathsf{Sign}|$. 
\begin{enumerate}
\item First, we define $\mathit{Tr}^n_v: \mathbf{Sen} (\Sigma) \to \mathbf{Sen} (\Sigma')$ for mapping first-order logic sentences into quantifier-free and ground first-order logic sentences. This will be done by turning existencial quantifiers into finite disjunctions over $\Sigma'$ in a way that the fresh new constant symbols $\{c_k\}_{1 \leq k \leq n}$ play the role of the only elements in the domain of interpretation. 

Let $v: \mathcal{X} \to \{c_k\}_{1 \leq k \leq n}$ be a function mapping the first-order variable symbols in $\mathcal{X}$ to constant symbols in $\{c_k\}_{1 \leq k \leq n}$, then $\mathit{Tr}^n_v$ is defined as follows\footnote{Let $v: \mathcal{X} \to \{c_k\}_{1 \leq k \leq n}$ be a function mapping the first-order variable symbols in $\mathcal{X}$ to constant symbols in $\{c_k\}_{1 \leq k \leq n}$, then $v\{x \mapsto c_i\}(y) = \left\{
\begin{array}{lr} 
v(y) & \text{; if $x \not= y$.}\\
c_i & \text{; if $x=y$.}
\end{array}
\right.$.}:
$$
\begin{array}{rcl}
\mathit{Tr}^n_v (P(t_1, \dots ,t_k)) & = & P(\mathit{Tr}^n_v (t_1), \dots, \mathit{Tr}^n_v (t_k)) \text{ , for all $P \in \{P_j\}_{j \in \mathcal{J}}$}\\
\mathit{Tr}^n_v (\alpha \vee \beta) & = & \mathit{Tr}^n_v (\alpha) \vee \mathit{Tr}^n_v (\beta)\\
\mathit{Tr}^n_v (\neg \alpha) & = & \neg \mathit{Tr}^n_v (\alpha)\\
\mathit{Tr}^n_v ((\exists x) \alpha) & = & \bigvee_{i=1}^{n} \mathit{Tr}^n_{v\{x \mapsto c_i\}} (\alpha)
\end{array}
$$
$$
\begin{array}{rcl}
\mathit{Tr}^n_v (x) & = & v (x) \text{ , for all $x \in \mathcal{X}$}\\
\mathit{Tr}^n_v (f (t_1, \dots, t_k)) & = & f (\mathit{Tr}^n_v (t_1), \dots, \mathit{Tr}^n_v (t_k)) \text{ , for all $f \in \{f_i\}_{i \in \mathcal{I}}$}
\end{array}
$$
\item Second, we define a partial function $\mathit{Tr}_\mathit{Int}: \mathbf{Sen} (\Sigma') \to \mathbf{Sen} (\Sigma')$ mapping quantifier-free and ground first-order sentences by interpreting the terms as the fresh new constant symbols $\{c_k\}_{1 \leq k \leq n}$.

Let $\alpha \in \mathbf{Sen} (\Sigma')$ and $\mathit{Term} (\alpha)$ be the set of terms mentioned in $\alpha$, excluding those in $\{c_k\}_{1 \leq k \leq n}$, then $\mathit{Tr}_\mathit{Int}$ is defined as follows:
$$\mathit{Tr}_\mathit{Int} (\alpha) = \bigvee_{f \in [\mathit{Term} (\alpha) \to \{c_k\}_{1 \leq k \leq n}]} \left[{\mathit{Tr}_\mathit{Int}}_f (\alpha) \land \bigwedge_{t \in \mathit{Term} (\alpha)} t = f(t)\right]$$
$$
\begin{array}{rcl}
{\mathit{Tr}_\mathit{Int}}_f (P (t_1, \dots , t_k)) & = & P (f(t_1), \dots, f(t_k)) \text{, for all $P \in \{P_j\}_{j \in \mathcal{J}}$},\\
{\mathit{Tr}_\mathit{Int}}_f  (t = t') & = & f(t) = f(t'),\\
{\mathit{Tr}_\mathit{Int}}_f  (\neg \alpha ) & = & \neg {\mathit{Tr}_\mathit{Int}}_f (\alpha),\\
{\mathit{Tr}_\mathit{Int}}_f  (\alpha \vee \beta) & = & {\mathit{Tr}_\mathit{Int}}_f  (\alpha) \vee {\mathit{Tr}_\mathit{Int}}_f  (\beta).
\end{array}
$$
\item Third, we define a partial function $\mathit{Tr}_\mathit{Prop}: \mathbf{Sen} (\Sigma') \to \mathbf{Sen}' (\gamma^{Sign} (\Sigma'))$ for mapping quantifier-free and ground first-order sentences, whose only terms are in $\{c_k\}_{1 \leq k \leq n}$, to propositional sentences, as follows:
$$
\begin{array}{rcl}
\mathit{Tr}_\mathit{Prop} (P (c_1, \dots , c_k)) & = & v_{\mbox{``$P(c_1, \ldots, c_k)$''}} \text{, for all $P \in \{P_j\}_{j \in \mathcal{J}}$},\\
\mathit{Tr}_\mathit{Prop} (c = c') & = & v_{\mbox{``$c = c'$''}},\\
\mathit{Tr}_\mathit{Prop} (\neg \alpha ) & = & \neg \mathit{Tr}_\mathit{Prop} (\alpha),\\
\mathit{Tr}_\mathit{Prop} (\alpha \vee \beta) & = & \mathit{Tr}_\mathit{Prop} (\alpha) \vee \mathit{Tr}_\mathit{Prop} (\beta).
\end{array}
$$
\end{enumerate}
Then $\gamma^{Sen}_{\Sigma} : \mathbf{Sen} (\Sigma) \to \mathbf{Sen}' (\gamma^{Sign} (\Sigma'))$ is defined as follows:\\
$$\gamma^{Sen}_{\Sigma}(\alpha) = \mathit{Tr}_\mathit{Prop} \(\mathit{Tr}_\mathit{Int} \(\mathit{Tr}^n_\emptyset (\alpha)\)\)$$
\end{definition}

In the previous definition $\gamma^{Sen}_{\Sigma}$ eliminates quantifiers replacing them with finite disjunctions and produces all the possible interpretations of the terms over the domain $\{c_1, \ldots, c_n\}$. Then, the last stage in the translation, replaces all atomic formulae (in which there are no term other than those in $\{c_1, \ldots, c_n\}$) by propositional variables labeled with the atomic formula they represent.

\begin{proposition}
\label{gamma-sen-proof}
$\gamma^{Sen}$ is a natural family of functions.
\end{proposition}
\begin{proof}
 The proof follows by first observing that given $\Sigma \in |\mathsf{Sign}|$, $\gamma^{Sen}_\Sigma$ (see Def.~\ref{gamma-sen}) is a function. Therefore, we need to prove that the following diagram commutes:
    \unitlength1cm
    \begin{center}
    \begin{picture}(11,3)(0,-3)
    \put(0,-3){}
    \put(0.3, -0.8){$\mathbf{Sen}(\Sigma_2)$}
    \put(0.8,-2.4){\vector(0,1){1.4}}
    \put(-0.4,-1.8){$\mathbf{Sen}(\sigma)$}
    \put(0.3,-2.8){$\mathbf{Sen}(\Sigma_1)$}
    \put(1.7,-0.7){\vector(1,0){2.4}}
    \put(2.6,-0.4){$\gamma^{Sen}_{\Sigma_2}$}
    \put(4.3,-0.8){$\mathbf{Sen}'(\gamma^{Sign} (\Sigma_2))$}
    \put(5.2,-2.4){\vector(0,1){1.4}}
    \put(5.3,-1.8){$\mathbf{Sen}'(\gamma^{Sign} (\sigma))$}
    \put(4.3,-2.8){$\mathbf{Sen}'(\gamma^{Sign} (\Sigma_1))$}
    \put(1.7,-2.7){\vector(1,0){2.4}}
    \put(2.6,-2.4){$\gamma^{Sen}_{\Sigma_1}$}
    
    \put(10, -0.8){$\Sigma_2$}
    \put(10.2,-2.4){\vector(0,1){1.4}}
    \put(10.4,-1.8){$\sigma$}
    \put(10,-2.8){$\Sigma_1$}
    \end{picture}
    \end{center}
where $\sigma : \Sigma_1 \to \Sigma_2 \in || \mathsf{Sign} ||$ is an homomorphism between two first-order signatures.
Note that:
 \begin{itemize}
     \item each $\gamma^{Sen}_{\Sigma}$ is the function mapping first-order $\Sigma$-formulae into propositional formulae,
     \item given $\sigma: \Sigma \to \Sigma' \in ||\mathsf{Sign}||$, $\mathbf{Sen}(\sigma)$ is a function mapping first-order $\Sigma$-formulae to first-order $\Sigma'$-formulae, and
     \item given $\sigma: \Sigma \to \Sigma' \in ||\mathsf{Sign}||$, $\mathbf{Sen}'(\gamma^{Sign} (\sigma))$ is a function mapping propositional $\gamma^{Sign} (\Sigma)$-formulae to propositional $\gamma^{Sign} (\Sigma')$-formulae.
 \end{itemize}
Then, the proof follows by observing that translating a first-order formulae across signatures and then mapping them to propositional formulae yields the same result as first mapping the first-order formula to a propositional one, and then translating the propositional formula.
\end{proof}

\begin{definition}
Let $n \in \Nat$, $\gamma^{Sign}: \mathsf{Sign} \to \mathsf{Sign}'$ be the functor of Def.~\ref{gamma-sign} and $\gamma^{Sen}: \mathbf{Sen} \to \mathbf{Sen}' \circ \gamma^{Sign}$ be the natural family of functions of Def.~\ref{gamma-sen}. Then, we define $\gamma^{\mathsf{Th}_0}: {\mathsf{Th}}_0 \to {\mathsf{Th}'}_0$ as:
$$\gamma^{\mathsf{Th}_0} (\<\Sigma, \Gamma\>) = \<\gamma^{Sign} (\Sigma), \{\gamma^{Sen}_\Sigma (\alpha) | \alpha \in \Gamma\}\>\ .$$
\end{definition}

\begin{proposition}
\label{gamma-th-proof}
    The functor $\gamma^{\mathsf{Th}_0}: {\mathsf{Th}}_0 \to {\mathsf{Th}'}_0$ is $\gamma^{Sen}$-sensible.
\end{proposition}
\begin{proof}
    It follows directly from the definition of $\gamma^{\mathsf{Th}_0}$ by observing it is explicitly constructed in terms of $\gamma^{Sign}$ and $\gamma^{Sen}$.
\end{proof}

The next definition provides the usual definition of model for propositional logic. A model is a function assigning a truth value from $\{\top, \bot\}$ to each extralogical (also referred to as rigid) symbol, appearing in the signature. It is easy to note that for such a model, not depending on any domain of discourse for interpreting objects, there is no possible notion of homomorphism that can be regarded as an arrow between two models thus, forcing us to consider them as organised as a discrete category.

\begin{definition}
Let $n \in \Nat$ and $\Sigma = \<\{f_i\}_{i \in \mathcal{I}}, \{P_j\}_{j \in \mathcal{J}}\> \in |\mathsf{Sign}|$. Then we define $\gamma^{Mod}_{\Sigma} : \mathbf{Mod}' (\gamma^{Sign} (\Sigma)) \to \mathbf{Mod} (\Sigma)$ as follows: for all $\mathit{val}: \gamma^{Sign} (\Sigma) \to \{\bot,\top\} \in |\mathbf{Mod}' (\gamma^{Sign} (\Sigma))|$, $\gamma^{Mod}_{\Sigma} (\mathit{val}) = \<\mathcal{S}, \mathcal{F}, \mathcal{P}\>$ such that:
\begin{itemize}
\item $\mathcal{S} = \{c_1, \ldots, c_n\}$,
\item $\mathcal{F} = \left\{ \overline{f} | f \in \{f_k\}_{k \in \mathcal{K}} \right\}$, where 

$\quad \overline{f} = \left\{\<c_1, \dots, c_k\> \mapsto c  \big| c_1, \ldots, c_k, c \in \mathcal{S}, \mathit{val}\left(v_{c = f(c_1, \dots, c_k)}\right) = \top\right\}$.
\item $\mathcal{P} = \left\{ \overline{P} | P \in \{P_j\}_{j \in \mathcal{J}} \right\}$, where 

$\quad \overline{P} = \left\{\<c_1, \ldots, c_k\>  \big| c_1, \ldots, c_k \in \mathcal{S}, \mathit{val}\left(v_{P(c_1, \ldots, c_k)}\right) = \top\right\}$, and
\end{itemize}
\noindent and for all $\mathit{Id}_{\mathit{val}}: {\mathit{val}} \to {\mathit{val}} \in ||\mathbf{Mod}' (\gamma^{Sign} (\Sigma))||$,  $\gamma^{Mod}_{\Sigma}(\mathit{Id}_{\mathit{val}}) = \mathit{Id}_{\gamma^{Mod}_{\Sigma}(\mathit{val})}$
\end{definition}

\begin{proposition}
\label{gamma-mod-proof}
    $\gamma^{Mod}_\Sigma$ is a functor. 
\end{proposition}
\begin{proof}
 $\gamma^{Mod}_\Sigma$ preserves identities and, since the source category is discrete, the compositions are trivially preserved.
\end{proof}

\begin{proposition}
    \label{gamma-mod-nt-proof}
$\gamma^{Mod}$ is a natural transformation and the functors $\gamma^{Mod}_\Sigma$ preserve the satisfaction condition.
\end{proposition}
\begin{proof}
    We need to prove that $\mathbf{Mod}(\sigma^\op) \circ \gamma^{Mod}_{\Sigma_2} = \gamma^{Mod}_{\Sigma_1} \circ \mathbf{Mod}' (\gamma^{\mathit{Sign}^\op} (\sigma^\op))$, with $\sigma : \Sigma_1 \to \Sigma_2 \in ||\mathsf{Sign}||$. It follows from observing that:
    \begin{itemize}
        \item $\mathbf{Mod}(\sigma^\op) (\mathcal{M})$ yields a first-order model obtained by capturing elements of $\mathcal{M}$ according to the signature morphism $\sigma$.
        \item $\mathbf{Mod}' (\gamma^{\mathit{Sign}^\op} (\sigma^\op)) (\mathcal{M}')$ yields a propositional valuation obtained by capturing the values in $\mathcal{M}'$ according to the signature morphism $\gamma^{\mathit{Sign}} (\sigma)$.
        \item $\gamma^{\mathit{Sign}}(\sigma)$, as it was defined in Def.~\ref{gamma-sign}, is the signature morphism in the category of signatures of propositional logic obtained from $\sigma$.
    \end{itemize}
    With these observations it can also be proved that the functor $\gamma^{Mod}_\Sigma$ preserves the satisfaction condition.
\end{proof}

\begin{proposition}
For all $n \in \Nat$, $\<\gamma^{\mathsf{Th}_0}, \gamma^{Sen}, \gamma^{Mod}\>: \mathbbm{I}_\mathit{FOL}(\mathcal{X}) \to \mathit{Prop}$ is a theoroidal comorphism between institutions. 
\end{proposition}
\begin{proof}
    The proof follows directly from Props.~\ref{gamma-sign-proof},~\ref{gamma-sen-proof},~\ref{gamma-th-proof},~\ref{gamma-mod-proof} and~\ref{gamma-mod-nt-proof}.
\end{proof}

Finally, the institution comorphism $\<\gamma^{Sign}, \gamma^{Sen}, \gamma^{Mod}\>$, together with $\gamma^{\mathsf{Th}_0}$, the theoroidal extension of $\gamma^{Sign}$, presented above, satisfy the hypothesis of Coro.~\ref{coro-sub-soundness-bmc-het}, thus providing an effective procedure, formalised in $\mathbbm{Q}_\text{Prop}$, for bounded counterexample finding that can be applied in the critical parts of a proof being developed within the proof calculus formalised in $\mathbbm{P}_\text{FOL}$.


\section{Conclusions}
\label{conclusions}
In this work we showed how effective satisfiability sub-calculi, a special type of satisfiability calculi, all of which were presented in \cite{lopezpombo:fi-166_4}, can be combined with proof calculi, as they were presented in the context of \emph{General logics} by Meseguer in \cite{meseguer:lc87}, in order to provide the foundations for methodological approaches to software analysis. This was done by relating, in an abstract categorical setting, the construction of counterexamples, using model finders, with the absence of proofs.

This methodology is based on the fact that searching of counterexamples is usually entangled with theorem proving in software analysis. As we mentioned in the preceding sections, there are many uses for counterexample finding capabilities, among which we find:
\begin{inparaenum}[1)]
\item gaining confidence on the correctness of the specification and the satisfaction of the property,
\item understanding the relevance of each hypothesis in the a proof, and
\item the analysis of the appropriateness of the addition of new hypothesis that might not be provable from the current set of hipothesis. 
\end{inparaenum}
This was exemplified by formalising a bounded counterexample finder as an effective satisfiability subcalculus for propositional logic, and then combining it with a proof calculus for first order logic with equality by means of a semantics preserving translation.\\

In \cite{gimenez:lafm13} we presented the tool \HeteroGenius\ as an implementation of a framework based on the (initially intuitive) notion of \emph{heterogeneous hybrid analysis}. The idea behind \HeteroGenius\ is to consider software analysis as a task developed by combining different techniques, following certain methodology. This can be done by considering an \emph{analysis structure} where nodes are judgements of the form $\Gamma \vdash_\Sigma \alpha$ and arrows relate judgements by applying a specific technique thus, providing some insight on its validity. The reader should note that the formalisation of both, proof calculi and satisfiability calculi, are not suitable for the implementation of analysis structures as the latter requires different tools to operate over the same structure, thus internalising their combination.

\bibliography{bibdatabase}
\bibliographystyle{splncs}

\appendix


\section{Selected proofs}
\label{proofs}
In this section we will present detailed explanations, definitions and proofs of the results supporting the examples we presented in Sections~\ref{preliminaries}~and~\ref{subsat}.

\subsection{Example~\ref{ex:sat-calculus}: Tableau method for first-order predicate logic}
\label{proofs-ejemplo-sat-calculus}
In Ex.~\ref{ex:sat-calculus} we presented the tableau method for first-order predicate logic and the intuitions for how it fits into the definition of a satisfiability calculus. In this section we will provide the formal definitions and the results proving it. Let $\mathbb{I}_{FOL} = \< \mathsf{Sign}, \mathbf{Sen}, \mathbf{Mod}, \{\models^{\Sigma}\}_{\Sigma \in |\mathsf{Sign}|} \>$, the institution of first-order predicate logic.

\begin{definition}
\label{def:str}
Let $\Sigma \in |\mathsf{Sign}|$ and $\Gamma \subseteq \mathbf{Sen}(\Sigma)$, we define $\mathit{Str}^{\Sigma, \Gamma} = \<\mathcal{O}, \mathcal{A}\>$ such that $\mathcal{O} = 2^{\mathbf{Sen}(\Sigma)}$ and $\mathcal{A} = \{\alpha:\{A_i\}_{i \in \mathcal{I}} \to \{B_j\}_{j \in \mathcal{J}}\ |\ \alpha=\{\alpha_j\}_{j \in \mathcal{J}}\}$, where for all $j \in \mathcal{J}$, $\alpha_j$ is a branch in a tableau for $\Gamma \cup \{B_j\}$ with leaves $\Delta \subseteq \{A_i\}_{i \in \mathcal{I}}$; $\Delta \models^\Sigma \Gamma \cup \{B_j\}$ follows as a direct consequence of the definition.
\end{definition}

\begin{lemma}
\label{lemma:str-category}
Let $\Sigma \in |\mathsf{Sign}|$ and $\Gamma \subseteq \mathbf{Sen}(\Sigma)$, then $\mathit{Str}^{\Sigma, \Gamma}$ defined as in Definition~\ref{def:str} is a category.
\end{lemma}
\begin{proof}
Let us prove that $\mathit{Str}^{\Sigma, \Gamma} = \<\mathcal{O}, \mathcal{A}\>$ is a category. For any set $\{A_i \}_{i \in \mathcal{I}} \in \mathcal{O}$, the identity is given by the collection of branches $\alpha_i : \{ A_i \}$ (of length $1$), i.e., no rule is applied. 

Now, given $\alpha : \{A_i\}_{i \in \mathcal{I}} \rightarrow \{ B_j \}_{j \in \mathcal{J}}, \beta : \{ B_j \}_{j \in \mathcal{J}} \rightarrow \{ C_q \}_{q \in \mathcal{Q}} \in \mathcal{A}$, their composition $\beta \circ \alpha = \gamma$ is defined as follows: let $\{\alpha_j : \{ B_j \} \rightarrow \dots \rightarrow S \cup \{ A_i \}_{i \in \mathcal{I}}\}_{j \in \mathcal{J}}$ and $\{\beta_q: \{ C_q \} \rightarrow \dots \rightarrow S' \cup  \{ B_j \}_{j \in \mathcal{J}}\}_{q \in \mathcal{Q}}$ be branches; then, $\{\gamma_q: \{C_q\} \rightarrow \dots \rightarrow \{B_j\}_{j \in \mathcal{J}} \cup S' \rightarrow \dots \rightarrow \{A_i\}_{i \in \mathcal{I}} \cup S \cup S'\}_{q \in \mathcal{Q}}$ is the branch obtained by extending each branch in $\beta$ with the corresponding branches in $\alpha$.

It remains to prove that $\circ$ has identities and is associative. Both proofs are straightforward by observing that $\circ$ is defined to be the concatenation of sequences of sets of formulae.
\end{proof}

\begin{lemma}
\label{lemma:strict-monoidal-cat}
Let $\Sigma \in |\mathsf{Sign}|$ and $\Gamma \subseteq \mathbf{Sen}(\Sigma)$; then $\<\mathit{Str}^{\Sigma, \Gamma}, \cup, \emptyset\>$, where $\cup: \mathit{Str}^{\Sigma, \Gamma} \times \mathit{Str}^{\Sigma, \Gamma} \to \mathit{Str}^{\Sigma, \Gamma}$ is the typical bi-functor on sets and functions, and $\emptyset$ is the neutral element for $\cup$, is a strict monoidal category. 
\end{lemma}
\begin{proof} Consider the bifunctor $\cup : \mathit{Str}^{\Sigma, \Gamma} \times \mathit{Str}^{\Sigma, \Gamma} \to \mathit{Str}^{\Sigma, \Gamma}$ which behaves as follows:
Given sets $A$ and $B$ $A \cup B$ is their union. Given a  pair of arrows $\alpha : A \rightarrow B$ and $\beta : C \rightarrow D$, where $\alpha = \{\alpha_i\}_{i \in \mathcal{I}}$
and $\beta = \{ \beta_j \}_{j \in \mathcal{J}}$, their union is $\alpha \cup \beta =  \{\alpha_i\}_{i \in \mathcal{I}} \cup  \{ \beta_j \}_{j \in \mathcal{J}}$. 
	Note that this functor is well defined: the union of the identities $id_{\{A_i\}_{i \in \mathcal{I}}} \cup id_{\{B_j\}_{j \in \mathcal{J}}}$ is a set of branches of length $1$
and so is an identity too; and the composition is preserved, since it is built point wise.
 	On the other hand,	the identity object of the monoidal category is $\emptyset$ and the natural isomorphisms are given by the identity which trivially makes the
	associativity and identity diagrams commute.
\end{proof}

\begin{definition}
\label{def-struct}
$\mathsf{Struct}_{SC}$ is defined as $\<\mathcal{O}, \mathcal{A}\>^\op$ where $\mathcal{O} = \{\<\mathit{Str}^{\Sigma, \Gamma}, \cup, \emptyset\> \ |\ \Sigma \in |\mathsf{Sign}| \land \Gamma \subseteq \mathbf{Sen}(\Sigma)\}$, and $\mathcal{A} = \{\widehat{\sigma}: \<\mathit{Str}^{\Sigma, \Gamma}, \cup, \emptyset\> \to \<\mathit{Str}^{\Sigma', \Gamma'}, \cup, \emptyset\>\ |\ \sigma:\<\Sigma, \Gamma\> \to \<\Sigma', \Gamma'\> \in ||\mathsf{Th}||\}$, the homomorphic extensions of the morphisms in $||\mathsf{Th}||$ to sets of formulae preserving the application of rules (i.e., the structure of the tableaux).
\end{definition}

\begin{lemma}
\label{lemma:struct-cat}
Let $\mathsf{Struct}_{SC}$ be defined as in Def.~\ref{def-struct}. Then, $\mathsf{Struct}_{SC}$ is a category.
\end{lemma}
\begin{proof}
First we prove that $\<\mathcal{O}, \mathcal{A}\>$ where $\mathcal{O} = \{\<\mathit{Str}^{\Sigma, \Gamma}, \cup, \emptyset\> \ |\ \Sigma \in |\mathsf{Sign}| \land \Gamma \subseteq \mathbf{Sen}(\Sigma)\}$, and $\mathcal{A} = \{\widehat{\sigma}: \<\mathit{Str}^{\Sigma, \Gamma}, \cup, \emptyset\> \to \<\mathit{Str}^{\Sigma', \Gamma'}, \cup, \emptyset\>\ |\ \sigma:\<\Sigma, \Gamma\> \to \<\Sigma', \Gamma'\> \in ||\mathsf{Th}||\}$ is a category. 

Morphisms $\widehat{\sigma} \in \mathcal{A}$ are the homomorphic extension of the morphisms $\sigma \in ||\mathsf{Th}||$ to the structure of the tableaux, translating sets of formulae and preserving the application of the rules. Following this, the composition of $\widehat{\sigma_1}, \widehat{\sigma_2} \in \mathcal{A}$, the homomorphic extension of $\sigma_1, \sigma_2 \in ||\mathsf{Th}||$, not only exists, but it is the homomorphic extension of the morphism $\sigma_1 \circ \sigma_2 \in ||\mathsf{Th}||$.  The associativity of the composition is also trivial to prove by considering that the morphisms are homomorphic extensions, and by the associativity of the composition of morphisms in $\mathsf{Th}$. The identity morphism is the homomorphic extension of the identity morphism for the corresponding signature.

Then, as a direct consequence we obtain that $\mathsf{Struct}_{SC}$ is a category.
\end{proof}

\begin{definition}
\label{def-m}
$\mathbf{M}: \mathsf{Th}^\op \to \mathsf{Struct}_{SC}$ is defined as $\mathbf{M} (\<\Sigma, \Gamma\>) = \<\mathit{Str}^{\Sigma, \Gamma}, \cup, \emptyset\>$ and for any $\sigma: \<\Sigma, \Gamma\> \to \<\Sigma', \Gamma'\> \in ||\mathsf{Th}||$, $\mathbf{M} (\sigma^\op) = \widehat{\sigma}^\op$ where $\widehat{\sigma}: \<\mathit{Str}^{\Sigma, \Gamma}, \cup, \emptyset\> \to \<\mathit{Str}^{\Sigma', \Gamma'}, \cup, \emptyset\>$ is the homomorphic extension of $\sigma$ to the structures in $\<\mathit{Str}^{\Sigma, \Gamma}, \cup, \emptyset\>$. 
\end{definition}

\begin{lemma}
\label{lemma:m-functor}
Let $\mathbf{M}: \mathsf{Th}^\op \to \mathsf{Struct}_{SC}$ be defined as in Definition~\ref{def-m}. Then $\mathbf{M}$ is a functor.
\end{lemma}
\begin{proof}
Let $id_{\<\Sigma, \Gamma\>}: \<\Sigma, \Gamma\> \to \<\Sigma, \Gamma\> \in ||\mathsf{Th}||$ be the identity morphism for $\<\Sigma, \Gamma\> \in |\mathsf{Th}|$. $\mathbf{M} ({\mathit{id}_{\<\Sigma, \Gamma\>}}^\op) = {id_{\<\mathit{Str}^{\Sigma, \Gamma}, \cup, \emptyset\>}}^\op$ because, by Def.~\ref{def-struct},  $\mathit{id}_{\<\mathit{Str}^{\Sigma, \Gamma}, \cup, \emptyset\>}$ is the homomorphic extension of $id_{\<\Sigma, \Gamma\>}$ to the structures in $\mathit{Str}^{\Sigma, \Gamma}$.

Let $\sigma_1: \<\Sigma_1, \Gamma_1\> \to \<\Sigma_2, \Gamma_2\>, \sigma_2: \<\Sigma_2, \Gamma_2\> \to \<\Sigma_3, \Gamma_3\> \in ||\mathsf{Th}||$; now, as composition of homomorphisms is a homomorphism, then $\mathbf{M} ({(\sigma_1 \circ \sigma_2)}^\op) = \mathbf{M} ({\sigma_2}^\op \circ {\sigma_1}^\op)$ by definition of opposite category. Thus, it is the composition $\mathbf{M} ({\sigma_2}^\op) \circ \mathbf{M} ({\sigma_1}^\op)$.
\end{proof}

%

\begin{definition}
\label{def-mods}
$\mathbf{Mods}: {\mathsf{Struct}_{SC}} \to \mathsf{Cat}$ is defined as:
\begin{itemize}
\item $\mathbf{Mods} (\<\mathit{Str}^{\Sigma, \Gamma}, \cup, \emptyset\>) = \<\mathcal{O}, \mathcal{A}\>$ where $\mathcal{O} = \bigcup\{ \<\Sigma, \widetilde{\Delta}\>\ |\ (\exists \alpha: \Delta \to \emptyset \in |\mathit{Str}^{\Sigma, \Gamma}|) (\widetilde{\Delta} \to \emptyset \in \alpha \land (\forall \alpha': \Delta' \to \Delta \in ||\mathit{Str}^{\Sigma, \Gamma}||)(\Delta' = \Delta) \land \neg(\exists \varphi)(\{\neg\varphi, \varphi\} \subseteq \widetilde{\Delta}))\}$ and $\mathcal{A} = \{ id_{T}: T \to T\ |\ T \in \mathcal{O} \}$, and
\item for all $\sigma: \<\Sigma, \Gamma\> \to \<\Sigma', \Gamma'\> \in ||\mathsf{Th}||$, $\mathbf{Mods}(\widehat{\sigma}^\op)(\<\Sigma, \delta\>) = \<\Sigma', \mathbf{Sen}(\sigma)(\delta)\>$. 
\end{itemize}
\end{definition}

\begin{lemma}
\label{lemma:mods-functor}
Let $\mathbf{Mods}: {\mathsf{Struct}_{SC}} \to \mathsf{Cat}$ defined as in Definition~\ref{def-mods}. Then, $\mathbf{Mods}$ is a functor.
\end{lemma}
\begin{proof}
As for each theory $\<\mathit{Str}^{\Sigma, \Gamma}, \cup, \emptyset\> \in |\mathsf{Struct}_{SC}|$, $\mathbf{Mods} (\<\mathit{Str}^{\Sigma, \Gamma}, \cup, \emptyset\>)$ is a discrete category containing theory presentations whose models are models of $\<\Sigma, \Gamma\>$, thus the only property that must be proved is that for all $\widehat{\sigma}: \<\mathit{Str}^{\Sigma, \Gamma}, \cup, \emptyset\> \to \<\mathit{Str}^{\Sigma', \Gamma'}, \cup, \emptyset\> \in ||\mathsf{Struct}_{SC}||$, $o \in |\mathbf{Mods}(\<\mathit{Str}^{\Sigma, \Gamma}, \cup, \emptyset\>)|$, $\mathbf{Mods}(\widehat{\sigma})(o) \in |\mathbf{Mods}(\<\mathit{Str}^{\Sigma', \Gamma'}, \cup, \emptyset\>)|$. By definition, $\mathbf{Mods}(\widehat{\sigma})(\<\Sigma, \widetilde{\Delta}\>) = \<\Sigma', \mathbf{Sen}(\sigma)(\widetilde{\Delta})\>$. Observe that, as a consequence of the fact that $\widehat{\sigma}$  is the homomorphic extension of $\mathbf{Sen}(\sigma)$ to the tree-like structure of tableaux, the theory presentation obtained by applying $\mathbf{Mods}(\widehat{\sigma})$ to a particular element of $\mathbf{Mods} (\<\mathit{Str}^{\Sigma, \Gamma}, \cup, \emptyset\>)$ is a theory presentation whose set of axioms is a leaf of a branch of a tableau in $\<\mathit{Str}^{\Sigma', \Gamma'}, \cup, \emptyset\>$. 
\end{proof}

\begin{definition}
\label{def-mu}
Let $\<\Sigma, \Delta\> \in |\mathsf{Th}|$, then we define $\mu_{\<\Sigma, \Delta\>}: \mathbf{models} (\<\Sigma, \Delta\>) \to \mathbf{Mod}_{FOL}(\<\Sigma, \Delta\>)$ as for all $\<\Sigma, \delta\> \in |\mathbf{models}(\<\Sigma, \Delta\>)|$, $\mu_{\<\Sigma, \Delta\>} (\<\Sigma, \delta\>) = \mathbf{Mod}_{FOL}(\<\Sigma, \delta\>)$.
\end{definition}

\begin{fact}
Let $\<\Sigma, \Gamma\> \in |\mathsf{Th}|$ and $\mu_{\<\Sigma, \Delta\>}: \mathbf{models} (\<\Sigma, \Delta\>) \to \mathbf{Mod}_{FOL}(\<\Sigma, \Delta\>)$ defined as in Def.~\ref{def-mu}. Let $\Sigma \in |\mathsf{Sign}_{FOL}|$ and $\Gamma \subseteq \mathbf{Sen}_{FOL}(\Sigma)$, then $\mu_{\<\Sigma, \Gamma\>}$ is a functor.
\end{fact}

\begin{lemma}
\label{lemma:mu-nattrans}
Let $\<\Sigma, \Gamma\> \in |\mathsf{Th}|$ and $\mu_{\<\Sigma, \Delta\>}: \mathbf{models} (\<\Sigma, \Delta\>) \to \mathbf{Mod}_{FOL}(\<\Sigma, \Delta\>)$ defined as in Definition~\ref{def-mu}. Then, $\mu$ is a natural family of functors.
\end{lemma}
\begin{proof}
Let $\<\Sigma, \Gamma\>, \<\Sigma', \Gamma'\> \in |\mathsf{Th}|$ and $\sigma: \<\Sigma, \Gamma\> \to \<\Sigma', \Gamma'\> \in |\mathsf{Th}_0|$. Then, the naturality condition for $\mu$ can be expressed in the following way:
\[
\xymatrix@C=13pt@R=20pt{  
\<\Sigma', \Delta'\> \ar@{->}@/^0pc/[dd]^{\mu_{\<\Sigma', \Delta'\>}} \ar@{->}@/^0pc/[rrrr]^{\mathbf{models} (\sigma)}
&
& &&
\<\Sigma, \Delta\> \ar@{->}@/^0pc/[dd]^{\mu_{\<\Sigma, \Delta\>}}\\
&
& &&\\
\mathbf{Mod}_{FOL}(\<\Sigma', \Delta'\>) \ar@{->}@/^0pc/[rrrr]^{\mathbf{Mod}_{FOL}(\sigma)}
&
& &&
\mathbf{Mod}_{FOL}(\<\Sigma, \Delta\>)\\
}
\]
It is trivial to check that this condition holds by observing that canonical models are closed theories, thus behaving as theory presentations in $\mathsf{Th}$.
\end{proof}

Now, from Lemmas~\ref{lemma:m-functor},~\ref{lemma:mods-functor},~and~\ref{lemma:mu-nattrans}, and considering the hypothesis that $\mathbb{I}_{FOL}$ is an institution, the following corollary follows.

\begin{corollary}
$\< \mathsf{Sign}, \mathbf{Sen}, \mathbf{Mod}, \{\models^{\Sigma}\}_{\Sigma \in |\mathsf{Sign}|}, \mathbf{M}, \mathbf{Mods}, \mu \>$ is a satisfiability calculus.
\end{corollary}

\subsection{Example~\ref{ex:eff-sat-subcalculus}: Effectiveness of the satisfiability subcalculus for finite presentations over the term-free restriction of first-order modal logic}
\label{proofs-ex:eff-sat-subcalculus}
In Example~\ref{ex:eff-sat-subcalculus} we presented an argument of how the satisfiability subcalculus of the term-free fragment of first-order modal logic of Ex.~\ref{ex:sat-subcalculus} fits the definition of an effective satisfiability subcalculus. In this section we will provide the formal definitions and the results proving it. Let us denote by $\mathbb{Q} = \< \mathsf{Sign}, \mathbf{Sen}, \mathbf{Mod},  \mathsf{Sign}_0, \mathsf{Sen}_0, \mathbf{ax},  \{\models^{\Sigma}\}_{\Sigma \in |\mathsf{Sign}|}, \mathbf{M}, \mathbf{Mods}, \mu \>$ a satisfiability subcalculus for first-order modal logic.

$\mathsf{Sign}_0$ is a complete subcategory of $\mathsf{Sign}$ so we assume $J: \mathsf{Sign}_0 \to \mathsf{Sign}$ to be the functor that for all $\Sigma \in |\mathsf{Sign}_0|$, $J (\Sigma) = \Sigma$ and for all $\sigma \in ||\mathsf{Sign}_0||$, $J (\sigma) = \sigma$.

\begin{definition}
$\mathbf{Sen}_0$ is defined as the subfunctor of $\mathbf{Sen}$ resulting from restricting the latter to the objects and morphisms in $\mathsf{Sign}_0$.
\end{definition}

\begin{lemma}
\label{ex:eff-sen0-space}
$\mathbf{Sen}_0: \mathsf{Sign}_0 \to \mathsf{Space}$ is a functor.
\end{lemma}
\begin{proof}
To prove that $\mathbf{Sen}_0: \mathsf{Sign}_0 \to \mathsf{Space}$ is a functor, we need to prove that:
\begin{inparaenum}[1.]
\item given $\Sigma \in |\mathsf{Sign}_0|$, $\mathbf{Sen}_0 (\Sigma)$ is a space, 
\item given $\sigma: \Sigma \to \Sigma' \in ||\mathsf{Sign}_0||$, $\mathbf{Sen}_0 (\sigma)$ is a total function between $\mathbf{Sen}_0 (\Sigma)$ and $\mathbf{Sen}_0 (\Sigma')$, and
\item $\mathbf{Sen}_0$ preserves identities and composition.
\end{inparaenum}

The first condition is trivial $\Sigma \in |\mathsf{Sign}_0|$, $\mathbf{Sen}_0 (\Sigma) = \mathbf{Sen} (\Sigma)$ which is the infinite set of finite formulae recognised by the regular grammar presented in Ex.~\ref{ex:sat-subcalculus} for first-order modal logic. Thus, $\mathbf{Sen}_0 (\Sigma)$ is a space. The second condition also results trivial because $\mathbf{Sen}_0 (\sigma)$ is the homomorphic extension of $\sigma$ along the grammar mentioned before, so it is a function mapping formulae in space $\mathbf{Sen}_0 (\Sigma)$ to formulae in space $\mathbf{Sen}_0 (\Sigma')$. Finally, it is easy to observe that whenever $\mathbf{Sen}_0$ is applied to an identity morphism, the result is an identity function between formulae of the corresponding space. The preservation of composition also follows easily by checking that the composition of the homomorphic extensions of two morphisms results in the same function that the homomorphic extension of the composition of the morphisms.
\end{proof}

\begin{lemma}
\label{ex:eff-forgetful}
Let $\mathcal{U}: \mathsf{Space} \to \mathsf{Set}$ be the obvious forgetful functor projecting the underlying set of objects of the space and the total functions between them as morphisms; then $\mathcal{U} \circ \mathbf{Sen}_0 = \mathbf{Sen} \circ J$.
\end{lemma}
\begin{proof}
The proof follows by observing that:
\begin{inparaenum}[1)] 
\item $J: \mathsf{Sign}_0 \hookrightarrow \mathsf{Sign}$ is the identity inclusion functor,
\item the nature of $\mathcal{U}: \mathsf{Space} \to \mathsf{Set}$, and
\item $\mathbf{Sen}_0$ is the subfunctor of $\mathbf{Sen}$ restricted to the objects and morphisms of $\mathsf{Sign}_0$
\end{inparaenum}

Let $\Sigma \in |\mathsf{Sign}_0|$, then, as $J$ is the identity inclusion functor, $J (\Sigma) = \Sigma$. Then, as $\Sigma \in |\mathsf{Sign}_0|$, $\mathbf{Sen} (\Sigma) =  \mathcal{U} (\mathbf{Sen}_0 (\Sigma))$ because $\mathbf{Sen}_0$ is the subfunctor of $\mathbf{Sen}$. The case of morphisms is analogous but considering sets, instead of functions.
\end{proof}

\begin{lemma}
\label{ex:eff-ax-space}
$\mathsf{ax}: \mathsf{Sign} \to \mathsf{Space}$ is a functor
\end{lemma}
\begin{proof}
The proof is analogous to the one of Lemma~\ref{ex:eff-sen0-space} but considering the extension of the functor $\mathbf{Sen}_0$, which operates on formulae, to finite sets of formulae.
\end{proof}

\begin{lemma}
\label{ex:eff-mods-space}
$\mathbf{Mods}: \mathsf{Struct}_{SC} \to \mathsf{Space}$ is a functor.
\end{lemma}
\begin{proof}
The proof is analogous to the one of Lemma~\ref{ex:eff-ax-space} but considering the extension of the functor $\mathbf{ax}$, which operates on finite sets formulae, to finite tree-like structures whose nodes are finite sets of formulae.
\end{proof}

\end{document}